\documentclass{ieeeaccess}
\usepackage{cite,color,url}
\usepackage{algorithm,algpseudocode}
\usepackage{amssymb,amsmath,mathtools}
\usepackage{amsfonts}
\usepackage{caption}
\usepackage{multirow}

\usepackage{bm,bbm}
\usepackage{booktabs}
\usepackage[caption=false,font=footnotesize]{subfig}
\graphicspath{{./img/}}

\usepackage{mathcomp}
\usepackage{amsthm}
\usepackage{newtxmath}
\theoremstyle{definition}
\newtheorem{theorem}{Proposition}

\algnewcommand\algorithmicinput{\textbf{Input:}}
\algnewcommand\Input{\item[\algorithmicinput]}

\def\BibTeX{{\rm B\kern-.05em{\sc i\kern-.025em b}\kern-.08em
    T\kern-.1667em\lower.7ex\hbox{E}\kern-.125emX}}

\begin{document}

\history{Date of publication xxxx 00, 0000, date of current version xxxx 00, 0000.}
\doi{10.1109/ACCESS.2017.DOI}

\title{Beamforming Feedback-based Model-Driven Angle of Departure Estimation Toward Legacy Support in WiFi Sensing:
    An Experimental Study}
\author{
    \uppercase{Sohei~Itahara}\authorrefmark{1}, \IEEEmembership{Student~Member,~IEEE},
    \uppercase{Sota~Kondo}\authorrefmark{1}, \IEEEmembership{Student~Member,~IEEE},
    \uppercase{Kota~Yamashita}\authorrefmark{1}, \IEEEmembership{Student~Member,~IEEE},
    \uppercase{Takayuki~Nishio}\authorrefmark{1,2}, \IEEEmembership{Senior~Member,~IEEE},
    \uppercase{Koji~Yamamoto}\authorrefmark{1}, \IEEEmembership{Senior~Member,~IEEE},
    \uppercase{and Yusuke~Koda}\authorrefmark{3}, \IEEEmembership{Member,~IEEE},
}
\address[1]{Graduate School of Informatics, Kyoto University, Kyoto 606-8501, Japan (e-mail: kyamamot@i.kyoto-u.ac.jp)}
\address[2]{School of Engineering, Tokyo Institute of Technology Ookayama, Meguro-ku, Tokyo, 152-8550, Japan (e-mail: nishio@ict.e.titech.ac.jp)}
\address[3]{Centre of Wireless Communications, University of Oulu, 90014 Oulu, Finland (e-mail: Yusuke.Koda@oulu.fi)}

\tfootnote{
    This research and development work was supported in part by the MIC/SCOPE \#JP196000002 and JSPS KAKENHI Grant Number JP18H01442.
}

\markboth
{Author \headeretal: Preparation of Papers for IEEE TRANSACTIONS and JOURNALS}
{Author \headeretal: Preparation of Papers for IEEE TRANSACTIONS and JOURNALS}

\corresp{Corresponding author: Koji Yamamoto (e-mail: kyamamot@i.kyoto-u.ac.jp).}

\begin{abstract}
    This study experimentally validated the possibility of angle of departure (AoD) estimation using multiple signal classification (MUSIC) with only WiFi control frames for beamforming feedback (BFF), defined in IEEE 802.11ac/ax.
    The examined BFF-based MUSIC is a model-driven algorithm, which does not require a pre-obtained database.
    This contrasts with most existing BFF-based sensing techniques, which are data-driven and require a pre-obtained database.
    Moreover, the BFF-based MUSIC affords an alternative AoD estimation method without access to channel state information (CSI).
    Specifically, the extensive experimental and numerical evaluations demonstrated that the BFF-based MUSIC successfully estimates the AoDs for multiple propagation paths.
    Moreover, the evaluations performed in this study revealed that the BFF-based MUSIC achieved a comparable error of AoD estimation to the CSI-based MUSIC,
    while BFF is a highly compressed version of CSI in IEEE 802.11ac/ax.
\end{abstract}


\titlepgskip=-15pt

\maketitle

\section{Introduction}
\label{sec:introduction}
WiFi sensing~\cite{yongsen2019wifi,zafari2019asurvey} is envisioned as a technology that adds value to existing wireless local area networks beyond the communication infrastructure.
In WiFi sensing, an example of widely used radio frequency (RF) information is channel state information (CSI),
which is measured in multiple-input multiple-output orthogonal frequency-division multiplexing (MIMO-OFDM) systems~\cite{yongsen2019wifi}.
CSI is generally measured in the MIMO-OFDM communication procedures and includes a high sensing capacity
to facilitate CSI-based sensing with low implementation cost and high sensing accuracy.

Presently, the next-generation WiFi standards task group, IEEE 802.11bf~\cite{11bf}, is actively embedding WiFi sensing ability to WiFi standards.
In IEEE 802.11bf~\cite{11bf},
it is required to allow WiFi sensing with legacy devices (i.e., devices whose physical (PHY) layers are compliant with legacy WiFi standards, such as IEEE 802.11ac/ax\cite{11ax,11ac}).
A challenge in meeting this requirement is that the legacy PHY layer processes and discards CSI, resulting in the disability of the CSI in WiFi sensing.

Beamforming feedback (BFF), which is a compressed version of CSI, has attracted attention as an alternative RF information to CSI, in order to address this challenge~\cite{murakami2018wireless,miyazaki2019init, takahashi2019dnn,kanda2021respiratory,kato2021csi2image,kudo2021deep,fukushima2019evaluating}.
Specifically, BFF includes a highly quantized right singular matrix of the CSI matrix for each subcarrier and subcarrier-averaged stream gain.
In IEEE 802.11ac/ax~\cite{11ac,11ax}, a station (STA) transmits BFFs to an access point (AP) without any encryption, allowing an arbitral WiFi device to obtain the BFFs with medium access control (MAC)-level frame-capturing tools.
Prior studies~\cite{murakami2018wireless,miyazaki2019init, takahashi2019dnn,kanda2021respiratory,kato2021csi2image,kondo2021bi,kudo2021deep,fukushima2019evaluating} have demonstrated the feasibility of BFF-based sensing in several sensing tasks, such as human localization and respiratory estimation.

However, the existing BFF-based sensing literature has lacked the following perspectives,
model-driven sensing and comparison of CSI to BFF in terms of sensing accuracy.
First, to the best of the authors' knowledge, in the BFF-based sensing literature, there are no model-driven algorithms, which geometrically estimate the surrounding environment based on mathematical modeling, although a vast of CSI-based model-driven algorithms~\cite{kotaru2015spotfi,xiong2013arraytrack} have been proposed.
In contrast, the existing BFF-based sensing methods~\cite{murakami2018wireless,miyazaki2019init, takahashi2019dnn,kanda2021respiratory,kato2021csi2image,kudo2021deep} are referred to as data-driven methods;
namely, the sensing tasks are conducted via pattern matching to a pre-obtained training dataset, which comprises the BFF and corresponding actual-measured target labels (e.g., human locations or device locations).
Because such training dataset generation procedure incurs tremendous human costs, the cost of data-driven sensing is generally higher than that of model-driven sensing.
Therefore, the lack of model-driven methods in the BFF-based sensing literature results in significant drawbacks to CSI-based sensing.

This motivated the development of a BFF-based model-driven sensing algorithm that does not require preparing the dataset.
To this end, we revisit model-based sensing in the CSI-based sensing literature.
A fundamental technique referred to as the multiple signal classification (MUSIC) algorithm~\cite{schmidt1986multiple} is used to estimate the angle of departure (AoD) for each of the multiple propagation paths.
Based on the original MUSIC algorithm~\cite{schmidt1986multiple}, which imposes some constraints,
there have been various extensions to CSI-based sensing, for example, the alleviation of constraints regarding antenna array~\cite{xinyu2019triangular} and propagation environment~\cite{kotaru2015spotfi}, and the realization of addressing-sensing tasks~\cite{li2017indotrack}.
These studies constructed high-capacity and widely applicable sensing frameworks.
However, whether the original MUSIC algorithm applies to BFF-based sensing remains unknown.

This paper presents model-driven analytics of BFF-based sensing and demonstrates that an extension of the MUSIC algorithm~\cite{schmidt1986multiple} can be realized using BFF.
Specifically, given the $\bar{\bm{\varLambda}}$ and $\bm{V}_k$ as the subcarrier-averaged stream gain and right singular matrix of CSI matrix at the $k$th subcarrier, the noise subspace vectors in the MUSIC algorithm are estimated as the eigenvectors of a covariance matrix $\sum_k \bm{V}_k\bar{\bm{\varLambda}}\bm{V}^\mathrm{H}_k$ with an eigenvalue of zero.
In contrast, CSI-based MUSIC generally uses a covariance matrix $\sum_k {\bm{h}_k}^\mathrm{H}\bm{h}_k$, where $\bm{h}_k$ is a row vector of the CSI matrix.
The mathematical analytics revealed that the role of the covariance matrix obtained from BFF has the same role as the covariance matrix obtained from CSI.
Our numerical evaluation and extensive experimental evaluations indicated that the BFF-based MUSIC algorithm accurately estimates AoDs and is comparable to the CSI-based MUSIC.

Second, to the best of the authors' knowledge, the existing BFF-based sensing approach has not provided a sensing-accuracy comparison between CSI and BFF.
Instead of the benefit of usability of BFF, because the BFF is a highly compressed version of CSI,
the sensing accuracy of the BFF is, in principle, lower than that of CSI.
Thus, the experimental comparisons of CSI and BFF are essential to assess the feasibility of replacing CSI with BFF.
We compared the AoD estimation accuracy of BFF- and CSI-based sensing and revealed that the BFF-based sensing achieves comparable accuracy to the CSI-based sensing.
Specifically, in three experimental environments, the median AoD estimation accuracy difference between BFF-based MUSIC and CSI-based MUSIC is smaller than 0.1\textdegree.

The contributions of this study are summarized as follows:
\begin{itemize}
    \item
          We analytically confirmed that the MUSIC algorithm can be performed using only the BFF frame.
          Specifically, using $\bar{\bm{\varLambda}}$ and $\bm{V}_k$, which are contained in the BFF frame,
          the noise subspace vectors in the MUSIC algorithm are estimated as the eigenvectors of $\sum_k \bm{V}_k\bar{\bm{\varLambda}}\bm{V}^\mathrm{H}_k$ with an eigenvalue of zero.
          This finding shows that the AoD estimation only using BFF is possible, shedding light on the applicability of model-driven BFF-based sensing to various sensing tasks (e.g., human sensing and device localization).
    \item
          We demonstrated the feasibility of model-driven BFF-based sensing as an alternative method without requiring access to CSI.
          More specifically, a numerical evaluation and extensive experimental evaluations reveal that, while the BFF procedure defined in IEEE 802.11ac/ax quantizes $\bm{V}_k$ and $\bar{\bm{\varLambda}}$ (e.g., $3\times 2$ complex matrix is represented by only 30\,bit), the AoD estimation accuracy of BFF-based MUSIC is comparable to that of CSI-based MUSIC.
          This is the first work that compares BFF-based sensing and CSI-based sensing in terms of sensing accuracy in the same experimental environment.
\end{itemize}

This study focused on the feasibility of the original MUSIC algorithm using BFF and the assessment of the accuracy degradation of BFF from CSI.
Thus, the comparison and implementation of more sophisticated CSI-based sensing methods such as \cite{kotaru2015spotfi} and data-driven BFF-based sensing methods are out of the scope of this study.

\subsection{Notations}
We denote the transpose of a matrix $\bm{H}$ as $\bm{H}^\mathrm{T}$, its conjugate as $\bm{H}^*$, its Hermitian transpose as $\bm{H}^\mathrm{H}$, and the $(i,j)$ element as $H_{i,j}$.
We denote the $i$th element of a vector $\bm{a}$ as $a_i$ and the Euclidian norm as $|\bm{a}|$.
The identity matrix is represented as $\bm{E}$.
The diagonal matrix, whose $i$th diagonal element is $a_i$, is represented as $\mathrm{diag}(\bm{a})$.
The $M\times N$ zero matrix is denoted as $\bm{0}_{M\times N}$.

\subsection{Related Works and Preliminary}
\subsubsection{Related Works}
Here, we provide a brief review of existing WiFi sensing literature, detailing the difference from such studies.

\begin{table*}[t!]
    \caption{
        Summary of BFF-based WiFi sensing.
    }
    \label{table:related_works}
    \centering
    \scalebox{1.}{
        \begin{tabular}{cccccccc}
            \toprule
                                                                                                               & Task                        & Model/Data driven?    & Comparison with CSI-based sensing & \\
            \midrule
            \cite{fukushima2019evaluating,murakami2018wireless,miyazaki2019init,takahashi2019dnn, kondo2021bi} & Human localization          & Data-driven           & No                                  \\
            \cite{kanda2021respiratory}                                                                        & Respiratory rate estimation & Model-driven          & No                                  \\
            \cite{kato2021csi2image}                                                                           & Camera image reconstruction & Data-driven           & No                                  \\
            \cite{kudo2021deep}                                                                                & Device localization         & Data-driven           & No                                  \\
            \cite{kondo2021bi}                                                                                 & AoD estimation              & Data-driven           & No                                  \\
            \textbf{This study}                                                                                & AoD estimation              & \textbf{Model-driven} & \textbf{Yes}                        \\
            \bottomrule
        \end{tabular}
    }
\end{table*}

\textbf{CSI-based sensing.}
Due to the rich sensing capacity of CSI, CSI has been attracted for providing RF information for WiFi sensing~\cite{yongsen2019wifi,zafari2019asurvey}.
There are various CSI-based sensing methods, including model-driven methods~\cite{yongsen2019wifi,zafari2019asurvey,xiong2013arraytrack,tian2016pila,kotaru2015spotfi} and data-driven methods~\cite{yongsen2019wifi,zafari2019asurvey,ahmed2020device}.
While the data-driven methods require a considerable cost to collect training datasets, model-driven methods are conducted without any training dataset, thus resulting in lower implementation costs.
A basic theorem of the model-driven methods is the MUSIC algorithm~\cite{xiong2013arraytrack,schmidt1986multiple,kotaru2015spotfi}, which is detailed in Section~\ref{sssec:MUSIC}.
Based on the MUSIC algorithm, various extended versions~\cite{xinyu2019triangular,kotaru2015spotfi, li2017indotrack} are proposed.

However, devices whose firmware is compliant with legacy WiFi standards (e.g., IEEE 802.11ac/ax) cannot conduct CSI-based sensing without remodeling their firmware.
This is because CSI is processed and discarded in the PHY layer at the legacy WiFi standards.
Thus, a remodeled firmware (e.g., \cite{gringoli2019free,halperin2011tool,xie2015precise}) is required to conduct CSI-based sensing.
Moreover, few wireless chips permit access to the PHY layer from the remodeled firmwares~\cite{gringoli2019free,halperin2011tool,xie2015precise}.
In contrast with CSI-based sensing, BFF-based sensing, which can be performed using arbitral devices compliant with the IEEE 802.11ac/ax, is the focus herein.

\textbf{BFF-based sensing.}
Table~\ref{table:related_works} summarizes the existing BFF-based sensing studies.
As mentioned in the previous section, because BFFs can be collected via the MAC-layer frame capturing without any constraints regarding the firmware,
BFF sensing has the potential as an alternative to CSI in WiFi sensing with legacy devices.
There are few studies on BFF-based sensing, for instance, those concerning human detection~\cite{murakami2018wireless, miyazaki2019init, takahashi2019dnn,kondo2021bi,fukushima2019evaluating}, device localization~\cite{kudo2021deep}, respiratory rate estimation\cite{kanda2021respiratory}, and camera image estimation~\cite{kato2021csi2image}.
Most of the studies~\cite{murakami2018wireless,miyazaki2019init,takahashi2019dnn,kato2021csi2image, kondo2021bi, kudo2021deep,fukushima2019evaluating} are categorized as data-driven methods.
Only \cite{kanda2021respiratory} is categorized into model-driven methods.
\cite{kanda2021respiratory} estimates the respiratory rate of a human by focusing on the relationship between the temporal variations of the BFF and respiratory rate.
However, \cite{kanda2021respiratory} is a heuristic and does not provide any propagation model-based analytics.
In contrast with these studies, the present study is based on a well-known propagation model~\cite{sayeed2002deconstructing} and analytically confirmed that the AoD is estimated using the BFF via MUSIC algorithm.

Moreover, this is the first work that presents accuracy comparisons between CSI and BFF.
Prior studies~\cite{murakami2018wireless, miyazaki2019init, takahashi2019dnn,kondo2021bi,kanda2021respiratory,kato2021csi2image, kudo2021deep, fukushima2019evaluating} have only provided the accuracy of BFF-based sensing and have not included comparisons between CSI- and BFF-based sensing.
Because the BFF includes significant quantization losses, the accuracy of the BFF-based sensing is principally inferior to that of the CSI-based sensing.
Thus, evaluating the degree of accuracy degradation is essential to assess whether BFF can be an alternative to CSI.
Our extensive experimental evaluations revealed that BFF-based MUSIC achieves a comparable median of AoD estimation accuracy to CSI-based MUSIC.

\subsubsection{Beamforming Feedback Scheme in 802.11ac/ax}
\label{sssec:BFF}
We consider a MIMO communication system in which a transmitter (e.g., AP) transmits signals to a receiver (e.g., STA).
We denote the CSI matrix from the transmitter to the receiver at the $k$th subcarrier as $\bm{H}_k\in \mathbb{C}^{M\times N}$, where $M$ and $N$ denote the number of antenna elements of the receiver and transmitter, respectively.
In IEEE 802.11ac/ax standards, to provide efficient eigen beam-space division multiplexing~\cite{miyashita2002high}, the receiver feedbacks the BFF frame to the transmitter~\cite{11ac,11ax}, which contains a compressed version of the CSI matrix.
Because the BFF frame is exchanged over the air without encryption,
BFFs can be obtained using the MAC frame-capturing tools, thus enabling an arbitral sniffer to perform BFF-based WiFi sensing without requiring access to the PHY layer components of the transmitter and receiver~\cite{miyazaki2019init}.

The BFF contains highly quantized right singular matrix $\bm{V}_k$ of the CSI matrix $\bm{H}_k$ for each subcarrier and a subcarrier-averaged stream gain~\cite{11ac,11ax}.
The right singular vector $\bm{V}_k$ is calculated using singular value decomposition as
\begin{align}
    \label{equ:SVD}
    \bm{H}_k = \bm{U}_k\bm{\varSigma}_k{\bm{V}_k}^\mathrm{H},
\end{align}
where $\bm{U}_k$ and $\bm{V}_k$ are unitary matrices, and $\bm{\varSigma}_k$ is a diagonal matrix with singular values~\cite{stewart1993early}.
Denoting number of subcarriers as $K$, the subcarrier-averaged stream gain is represented by a diagonal matrix $\bar{\bm{\varLambda}}$, where
\begin{align}
    \bar{\bm{\varLambda}} = \frac{1}{K}\sum_{k=1}^{K}{\bm{\varSigma}_k}^2.
\end{align}
Notably, the diagonal elements of $\bm{\varSigma}_k$ are generally real and positive and are listed in descending order.

As per IEEE 802.11ac/ax~\cite{11ac,11ax} standards, the $\bm{V}_k$ and $\bar{\bm{\varLambda}}$ are highly quantized to reduce the payload size of the BFF frame.
Specifically, $\bm{V}_k$ is converted to the $N^\mathrm{angle}$ angles without any quantization losses using Givens rotation, where $N^\mathrm{angle}$ is determined by $N$ and $M$.
The $N^\mathrm{angle}$ angles are quantized with a predefined quantization step size $\Delta$ and contained in a BFF frame.
The IEEE 802.11ac~\cite{11ac} defines four quantization step sizes, namely $\pi/4$\,rad, $\pi/8$\,rad, $\pi/16$\,rad, and $\pi/32$\,rad.
The subcarrier-averaged stream gain $\bar{\bm{\varLambda}}$ are quantized with quantize step sizes of 0.25\,dB~\cite{11ac}.

\subsubsection{Propagation Model}
We consider a discrete physical propagation model~\cite{sayeed2002deconstructing}, wherein a uniform linear array is employed at the transmitter and receiver.
In the following description, for simplicity, we assume that the distances between consecutive antennas at the transmitter and receiver are the same, which is denoted as $d$.\footnote{This assumption can be easily expanded to the case that the distances between consecutive antennas differs between the transmitter and receiver.}

Let $L$ be the number of propagation paths.
Additionally, let $\phi_l$ be the AoD and $\theta_l$ be the angle of arrival of the $l$th path.
The complex scalar $r_l$ denotes the attenuation from the transmitter's first antenna to the receiver's first antenna by the signal traveling along the $l$th propagation path.
We denote a complex phase shift $a(\theta)$ as $\exp(2\pi d\sin (\theta)/\lambda)$, where $\lambda$ is the wavelength.
For shorthand notation, let $L$-dimensional vectors $\bm{\theta}$, $\bm{\phi}$, and $\bm{r}$ represent $(\theta_1,\ldots, \theta_L)^\mathrm{T}$, $(\phi_1,\ldots, \phi_L)^\mathrm{T}$, and $(r_1,\ldots, r_L)^\mathrm{T}$, respectively.
Additionally, we denote the steering vector $\bm{a}(\theta) \coloneqq (1, a(\theta), \ldots, a(\theta)^{M-1})^\mathrm{T}$;
$L \times M$ steering matrix $\bm{A}(\bm{\theta})\coloneqq(\bm{a}(\theta_1),\ldots,\bm{a}(\theta_L))^\mathrm{T}$; and
$L \times L$ diagonal matrix $\bm{R}\coloneqq\mathrm{diag}(\bm{r})$.
In the discrete physical propagation model~\cite{sayeed2002deconstructing}, the CSI matrix $\bm{H}$ is represented as
\begin{align}
    \label{equ:csi}
    \bm{H} = \bm{A}(\bm{\theta})\bm{R}\bm{A}(\bm{\phi})^\mathrm{H}.
\end{align}

\subsubsection{Multiple Signal Classification (MUSIC) Algorithm}
\label{sssec:MUSIC}
The CSI-based MUSIC algorithm~\cite{xiong2013arraytrack} estimates multiple AoDs from CSI by assuming $L<N$.
The general CSI-based MUSIC consists of three steps as follows~\cite{xiong2013arraytrack,schmidt1986multiple,kotaru2015spotfi}.
First, given an arbitral slim and full-rank matrix as $\bm{S}$,
we estimate a matrix $\bm{X}$ represented by $\bm{S}\bm{A}(\bm{\phi})^\mathrm{H}$.
For example, in \cite{schmidt1986multiple,kotaru2015spotfi}, the matrix $\bm{X}_0$ is a $K\times M$ matrix, whose $k$th row vector is the first row vector of the CSI matrix at the $k$th subcarrier.

Considering the propagation model denoted in \eqref{equ:csi}, the first row vector of the CSI matrix at the $k$th subcarrier is represented by
\begin{align}
    \bm{h}_k = \left(\bm{r}_k^\mathrm{T} \bm{A}(\bm{\phi})^\mathrm{H} \right)^\mathrm{T}.
\end{align}
Given $K\times L$ matrix $\bm{S}_0$ as $(\bm{r}_1, \ldots, \bm{r}_K)^\mathrm{T}$, the matrix $\bm{X}_0$ is represented by
\begin{align}
    \bm{X}_0 & = \left(\bm{h}_1,\ldots,\bm{h}_1\right)^\mathrm{T} =
    \left(\bm{A}(\bm{\phi})^\mathrm{*} \bm{r}_1, \ldots,  \bm{A}(\bm{\phi})^\mathrm{*} \bm{r}_K \right)^\mathrm{T}            \\
             & =(\bm{r}_1, \ldots, \bm{r}_K)^\mathrm{T} \bm{A}(\bm{\phi})^\mathrm{H} = \bm{S}_0 \bm{A}(\bm{\phi})^\mathrm{H}.
\end{align}
Generally, $\bm{S}_0$ is slim and full-rank~\cite{xiong2013arraytrack,schmidt1986multiple}; thus, $\bm{X}_0$ is represented as a product of the slim and full-rank matrix and $\bm{A}(\bm{\phi})^\mathrm{H}$.

Second, a covariance matrix $\bm{C}\coloneqq \bm{X}^\mathrm{H}\bm{X}$ is obtained, and the $M-L$ noise subspace vectors $\bm{e}_1,\ldots,\bm{e}_{M-L}$ are calculated as the $M-L$ eigenvectors of $\bm{C}$ with small eigenvalues.
Lastly, the AoDs are estimated as angles that achieve peaks of MUSIC spectrum $g(\phi)$, where
\begin{align}
    g(\phi) \coloneqq \frac{1}{\bm{a}(\phi)^\mathrm{H}{\bm{E}_\mathrm{N}}^{\mathrm{H}}\bm{E}_\mathrm{N}\bm{a}(\phi)},
\end{align}
where $\bm{E}_\mathrm{N} = (\bm{e}_1,\ldots,\bm{e}_{M-L})$.

\section{Beamforming Feedback-based Multiple Signal Classification}
\label{sec:prop}
Fig.~\ref{fig:system} shows the system model consisting of an STA, an AP, and a sniffer.
The STA receives the sounding frame (e.g., the null data packet in IEEE 802.11ac/ax~\cite{11ac,11ax}) from the AP, estimates the CSI, and calculates the BFF from the CSI, which is detailed in Section~\ref{sssec:BFF}.
Then, the STA transmits the BFF to the AP without any encryption.
The sniffer captures the BFF transmitted from the STA, decodes the BFF, and obtains the right singular matrix $\bm{V}_k$ for each subcarrier and subcarrier-averaged stream gain $\bar{\bm{\varLambda}}$.
Subsequently, the sniffer estimates the AoDs of the AP using the BFF-based MUSIC method, which is detailed in the following sections.

\begin{figure}[!t]
    \centering
    \includegraphics[width=0.44\textwidth]{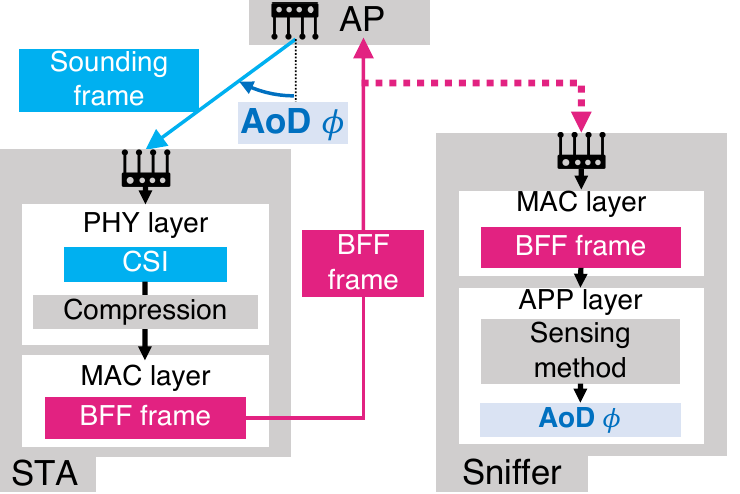}
    \caption{System model of BFF-based MUSIC.
        STA transmits BFF to AP without any encryption, allowing the sniffer to capture the BFF and conduct BFF-based sensing.
    }
    \label{fig:system}
\end{figure}

This study confirmed that the MUSIC algorithm is applicable using only the BFF to estimate multiple AoDs, which is proved in Proposition~\ref{tho:1}.
Specifically, assuming that $\bar{\bm{\varLambda}}={\bm{\varSigma}_k}^2$ for all $k$, the covariance matrix used in the MUSIC algorithm is estimated as
\begin{align}
    \label{equ:coeff}
    \bm{C} = \sum_{k=1}^K {\bm{V}_k}\bar{\bm{\varLambda}}\,{\bm{V}_k}^\mathrm{H}.
\end{align}
Based on the covariance matrix\footnote{
    It should be noted that in \cite{kudo2021deep}, \eqref{equ:coeff} is used for the feature extraction method in the data preprocessing procedure for data-driven BFF-based sensing.}, the AoDs are estimated by the general MUSIC algorithm, which is detailed in Section~\ref{sssec:MUSIC}.

\begin{theorem}
    \label{tho:1}
    Given a slim and full-rank matrix $\bm{S}$ and assuming $\bar{\bm{\varLambda}}=\bm{\varSigma}_k^2$,
    the covariance matrix $\bm{C}$ defined in \eqref{equ:coeff} is denoted by a covariance matrix of $\bm{S}\bm{A}(\bm{\phi})^\mathrm{H}$.
\end{theorem}

\begin{proof}
    Using the aforementioned assumptions, $\bm{C}$ in \eqref{equ:coeff} is expressed as
    \begin{align}
        \bm{C} = \sum_{k=1}^K {\bm{V}_k}\bm{\varSigma}_k^2{\bm{V}_k}^\mathrm{H}.
    \end{align}
    Substituting \eqref{equ:SVD} and \eqref{equ:csi} to $\bm{C}$, we obtain
    \begin{align}
        \bm{C} & =\sum_{k=1}^K {\bm{H}_k}^\mathrm{H}{\bm{H}_k}
        = \sum_{k=1}^K \bm{A}(\bm{\phi}){\bm{R}_k}^\mathrm{H}\bm{A}(\bm{\theta})\bm{A}(\bm{\theta})^\mathrm{H}\bm{R}_k\bm{A}(\bm{\phi})^\mathrm{H} \notag                   \\
               & =  \bm{A}(\bm{\phi})\left(\sum_{k=1}^K {\bm{R}_k}^\mathrm{H}\bm{A}(\bm{\theta})\bm{A}(\bm{\theta})^\mathrm{H}\bm{R}_k \right) \bm{A}(\bm{\phi})^\mathrm{H}
    \end{align}
    Using $KM\times M$ matrix $\bm{S} \coloneqq \left({\bm{R}_1}^\mathrm{H}\bm{A}(\bm{\theta}),\ldots,{\bm{R}_K}^\mathrm{H}\bm{A}(\bm{\theta}) \right)^\mathrm{H}$,
    \begin{align}
        \label{equ:theo1}
        \bm{C}
        = \bm{A}(\bm{\phi})\bm{S}^\mathrm{H}\bm{S}\bm{A}(\bm{\phi})^\mathrm{H}.
    \end{align}

    Thus, this proposition essentially proves that $\bm{S}$ is full-rank.
    Based on the deduction that $\bm{S}_0$ is slim and full-rank, which is denoted in Section~\ref{sssec:MUSIC}, the above proposition is proved by indirect proof.
    We denote a diagonal matrix $\hat{\bm{A}}(\bm{\theta})$ as $\mathrm{diag}(a(\theta_1),\ldots,a(\theta_L))$.
    If $\bm{S}$ is not a full-rank matrix, a non-zero vector $\bm{x} \in \mathbb{C}^{L}$ satisfies
    \begin{align}
        \label{equ:contra}
        \bm{S}\bm{x}=\bm{0}_{KM\times 1}.
    \end{align}
    The equation~\eqref{equ:contra} is equivalent to that, for all $m=1,\ldots,M$, $\bm{x}$ satisfies
    \begin{align}
        \bm{S}_0\hat{\bm{A}}(\bm{\theta})^{m-1} \bm{x}=\bm{0}_{K\times 1}.
    \end{align}
    However, as denoted in~\cite{kotaru2015spotfi}, $\bm{S}_0$ is generally slim and full-rank, and $\hat{\bm{A}}(\bm{\theta})^{m-1}$ is regular;
    thus, this shows a contradiction.
\end{proof}

\subsection{Detail Procedure}
The detailed procedure of the BFF-based MUSIC algorithm is presented in Algorithm~\ref{alg:B_MUSIC}.
The STA estimates the CSI using a sounding frame transmitted from the AP, calculates the BFF from the CSI, and transmits the BFF to the AP.
We denote $\bm{H}_{k,i}$ as the CSI matrix at the $k$th subcarrier from the $i$th sounding frame.
We also denote the right singular matrix and subcarrier-averaged stream gain of $\bm{H}_{k,i}$ as $\bm{V}_{k,i}$ and $\bar{\bm{\varLambda}}_{i}$, respectively.
The BFF corresponding to the $i$th sounding frame includes $\left(\bm{V}_{1,i},\ldots, \bm{V}_{K,i}\right)$ and $\bar{\bm{\varLambda}}_i$.
The $\bar{\bm{\varLambda}}_i$ and $\bm{V}_{k,i}$ include quantization errors because the BFF frame is highly quantized in IEEE 802.11ac/ax~\cite{11ac,11ax}, which is detailed in Section~\ref{sssec:BFF}.

The frame capture obtains $N^\mathrm{pct}$ BFF frames transmitted from the STA and estimates multiple AoDs of the AP from the BFF frames.
For each captured BFF frame, the frame capture obtains subcarrier-averaged stream gain $\bar{\bm{\varLambda}}_i$ and the right singular matrix $\bm{V}_{k,i}$.
Using ${\bar{\bm{\varLambda}}_i}$ and ${\bm{V}_{k,i}}$, the covariance matrix $\bm{C}_i$ is calculated as
\begin{align}
    \bm{C}_i = \frac{1}{K} \sum_{k=1}^K \bm{W}_k{\bm{V}_{k,i}}\bar{\bm{\varLambda}}_i\,{\bm{V}_{k,i}}^\mathrm{H}\bm{W}_k^\mathrm{H},
\end{align}
where $\bm{W}_k$ is a diagonal matrix that compensates for the phase sift introduced at the AP.
The methods to estimate $\bm{W}_k$ are detailed in Section~\ref{ssec:calib}.
We average $\bm{C}_i$ among $N^\mathrm{pct}$ packets and use the averaged covariance matrix $\bm{C}^\mathrm{ave}$ in the following MUSIC procedure, where
\begin{align}
    \bm{C}^\mathrm{ave} = \frac{1}{N^\mathrm{pct}} \sum_{i=1}^{N^\mathrm{pct}}\bm{C}_i.
\end{align}
Following the existing CSI-based MUSIC methods~\cite{shan1985onspatial,xiong2013arraytrack}, we adopt spatial smoothing to $\bm{C}^\mathrm{ave}$.
We denote the spatial smoothing function as $f^\mathrm{smt}$ and the smoothed covariance matrix as $\bm{C}^\mathrm{smt}$, where $\bm{C}^\mathrm{smt} = f^\mathrm{smt}(\bm{C}^\mathrm{ave})$.
The spatial smoothing procedure is detailed in Section~\ref{ssec:smt}.
From the smoothed covariance matrix $\bm{C}^\mathrm{smt}$, we estimate AoDs using the general MUSIC algorithm~\cite{schmidt1986multiple}, as described in Section~\ref{sssec:MUSIC}.

Notably, the estimation of the number of the propagation paths $L$ is required in the BFF-based MUSIC algorithm as with the CSI-based MUSIC algorithm.
In this work, we assume that $L$ is given, and the number of path estimation problems is out-of-scope.
This is because the problem is not specific to BFF-based sensing.

\begin{algorithm}[t]
    \caption{BFF-based MUSIC}
    \label{alg:B_MUSIC}
    \begin{algorithmic}[1]
        \Input $N^\mathrm{pct}$ BFF frames
        \For{each packet $i$}
        \State $\bm{C}_i = \frac{1}{K} \sum_{k=1}^K \bm{W}_k{\bm{V}_{k,i}}\bar{\bm{\varLambda}}_i\,{\bm{V}_{k,i}}^\mathrm{H}\bm{W}_k^\mathrm{H}$.
        \EndFor
        \State Averaging among packets: $\bm{C}^\mathrm{ave} = \frac{1}{N^\mathrm{pct}}\sum_{n=1}^{N^\mathrm{pct}} \bm{C}_n$
        \State Spatial smoothing: $\bm{C}^\mathrm{smt} = f^\mathrm{smt}(\bm{C}^\mathrm{ave})$
        \State Obtain eigenvectors $\bm{e}_1,\ldots,\bm{e}_{M}$ of $\bm{C}^\mathrm{smt}$, where $\bm{e}_1,\ldots,\bm{e}_{M}$ is aligned in descending order of its eigenvalue.
        \State Calculate noise subspace matrix $\bm{E}_\mathrm{N} = (\bm{e}_{1},\ldots,\bm{e}_{M-L})^\mathrm{T}$.
        \State Evaluate MUSIC spectrum $1/\bm{a}(\phi)^\mathrm{H}{\bm{E}_\mathrm{N}}^{\mathrm{H}}\bm{E}_\mathrm{N}\bm{a}(\phi)$.
        \State Obtain AoDs as $L$ peaks of MUSIC spectrum.
    \end{algorithmic}
\end{algorithm}

\subsection{Calibration Procedure}
\label{ssec:calib}
To provide accurate AoD estimation, the compensation for the phase offset introduced at the AP is required~\cite{xiong2013arraytrack}.
To this end, we implemented calibration method that estimates the phase shift difference between the antenna elements.
The calibration procedure measures the BFF at the environment where the number of propagation paths is only one and the AoD is given; subsequently, the phase offset at the AP is estimated.
Specifically, the calibration procedure is as follows: the covariance matrix of the CSI matrix is estimated from the BFF; and the eigenvector of the covariance matrix with the largest eigenvalue corresponds to the phase shift of the AP.

Formally, we denote the phase offset introduced at the $n$th antenna of the AP to be $\mathrm{e}^{\mathrm{j}\tau_{n,k}}$.
The calibration procedure estimates $\mathrm{e}^{\mathrm{j}(\tau_{n,k}-\tau_{1,k})}$.
For shorthand notation, we denote a diagonal matrix $\bm{W}_k$ as
\mbox{$\mathrm{diag}\!\left(1, \mathrm{e}^{\mathrm{j}(\tau_{2,k}-\tau_{1,k})},\ldots, \mathrm{e}^{\mathrm{j}(\tau_{N,k}-\tau_{1,k})}\right)$}.
Considering the $N\times 1$ MIMO system, and given that $L=1$ and the pre-obtained AoD is $\hat{\phi}$,
the observed CSI matrix is denoted as
\begin{align}
    \bm{H}^\mathrm{obs}_k = \mathrm{e}^{\mathrm{j}\tau_{1,k}} r_k \bm{a}(\hat{\phi})\,\bm{W}_k,
\end{align}
where $r_k$ denotes the complex path gain.

The calibration procedure estimates $\bm{W}_k$ using the pre-obtained AoD $\hat{\phi}$ and BFF calculated from $\bm{H}^\mathrm{obs}_k$ as follows.
We denote the right singular matrix and subcarrier-averaged stream gain of $\bm{H}^\mathrm{obs}_k$ as $\bm{V}_k^\mathrm{obs}$ and $\bar{\bm{\varLambda}}^\mathrm{obs}$, respectively.
First, the covariance matrix of $\bm{H}^\mathrm{obs}_k$ is estimated as $\bm{V}_k^\mathrm{obs} \bar{\bm{\varLambda}}^\mathrm{obs} (\bm{V}_k^\mathrm{obs})^\mathrm{H}$.
The covariance matrix is also represented by
\begin{align}
    \label{equ:cal}
    (\bm{H}^\mathrm{obs}_k)^\mathrm{H} \bm{H}^\mathrm{obs}_k =  |r_k|^2 {\bm{W}_k}^\mathrm{H} \bm{a}(\hat{\phi})^\mathrm{H} \bm{a}(\hat{\phi})\,\bm{W}_k.
\end{align}
From \eqref{equ:cal}, the covariance matrix has $N-1$ eigenvectors with an eigenvalue of zero and an eigenvector with an eigenvalue of $|r_k|^2$, and the latter eigenvector is ${\bm{W}_k}^\mathrm{H} \bm{a}(\hat{\phi})^\mathrm{H}$.
Thus, denoting the latter eigenvector as $\bm{x} \coloneqq (1, x_2,\ldots,x_N)^\mathrm{T}$,
$\bm{W}_k$ is estimated as
\begin{align}
    \bm{W}_k = \mathrm{diag}(\bm{a}(\hat{\phi}))^\mathrm{H}\mathrm{diag}(\bm{x})^\mathrm{H}.
\end{align}
In the MUSIC algorithm, which is implemented after the calibration,  $\bm{W}_k\bm{V}_k$ is used instead of $\bm{V}_k$.

\subsection{Spatial Smoothing}
\label{ssec:smt}
As denoted in \cite{shan1985onspatial,xiong2013arraytrack}, when the multipath signals are phase-synchronized with each other,
the distinct multipath signals are recognized as one superposed signal, resulting in false peaks in the MUSIC spectrum.
To address the problem, we adopt spatial smoothing~\cite{shan1985onspatial,xiong2013arraytrack},
which splits the AP's antenna array into multiple sub-antenna arrays.
Given that $M'$ antennas are integrated into a sub-antenna array,
the antenna array with $M$ antennas are considered $M-M'+1$ sub-antenna arrays.
The covariance matrix is calculated for each sub-antenna array in the spatial smoothing procedure, and the covariance matrices are averaged.
Specifically, given the covariance matrix for the $j$th sub-antenna array as ${\bm{C}^\mathrm{sub}}_{j}\in \mathbb{C}^{M'\times M'}$, ${\bm{C}^\mathrm{sub}}_{j}$ is a submatrix of $\bm{C}$, where $1,\ldots,j-1, j+M',\ldots,M$ rows and columns are removed from $\bm{C}$.
The averaged covariance matrix $\bm{C}^\mathrm{smt} \in \mathcal{C}^{M'\times M'}$ is obtained as
\begin{align}
    \bm{C}^\mathrm{smt} = \frac{1}{M-M'+1}\sum_{j=1}^{M-M'+1} \bm{C}^\mathrm{sub}_{j}.
\end{align}
The averaged covariance matrix $\bm{C}^\mathrm{smt}$ is used for estimating the noise subspace vectors, instead of the original covariance matrix $\bm{C}$.

\section{Numerical Evaluation}
\label{sec:num_evaluation}
Because the ground-truth multiple AoDs generally cannot be measured in a real-world environment, we examined the capacity of the BFF-based MUSIC to estimate multiple AoDs using a numerical evaluation.
Moreover, in the extensive experimental evaluations in real-world environments provided in Section~\ref{sec:exp_evaluation},
we evaluated the accuracy of the AoD estimation, assuming that only the direct path exists.

\subsection{Setup}
\label{ssec:nume_setup}
Fig.~\ref{fig:setup} illustrates the system, which comprises an AP, an STA, and a reflection point,
resulting in two different propagation paths between an antenna element of the AP and that of the STA---a direct path and an indirect path caused by the reflection point.
The STA and reflection point exist at (0\,m, 10\,m) and (5.5\,m, 3\,m), respectively, whereas the AP exists at either of 11 points on the x-axis.
Specifically, the $n_\mathrm{a}$th AP's position is denoted by $(n_\mathrm{a}-5\,\mathrm{m},0\,\mathrm{m})$, where $0\leq n_\mathrm{a} \leq 10$.
The AP and STA are equipped with uniform array antennas.
Each of the antenna arrays contains four antenna elements that are parallel to the x-axis.

We assume free-space propagation, wherein the indirect paths are decayed by 0.3 of the amplitude, and ignore the effect of the reflection more than once.
The CSI estimation is emulated by adding Gaussian noise to ground-truth CSI matrix $\bm{H}_k$.
Specifically, the estimated CSI at the $k$th subcarrier is denoted as
\begin{align}
    \bm{H}^\mathrm{obs}_k = \bm{H}_k + \bm{N},
\end{align}
where $\bm{N}$ is an $M \times N$ complex matrix whose real and imaginary parts of the elements follow a Gaussian distribution with mean 0 and variance $\sigma^2/2$.
It should be noted that $\sigma^2$ is the noise power at each antenna element.
We calculate $\bm{H}^\mathrm{obs}_k$ for each subcarrier $k$ and then obtain the $\bm{V}_k$ for each subcarrier and subcarrier-averaged stream gain $\bar{\bm{\varLambda}}$, by following the procedure denoted in Section~\ref{sssec:BFF}.
Specifically, we select the quantization step size $\Delta$ of $\pi/32$\,rad for the quantization of $\bm{V}_k$,\footnote{The quantization step size $\Delta$ of $\pi/32$\,rad is one of the quantization step sizes defined in IEEE 802.11ac and used in the BFF procedure in commercial APs, ASUS RT-AC86U and Buffalo WXR-5700AX7S.\label{foot:qss}}
resulting in the $4\times 4$ right singular matrix $\bm{V}_k$ being represented by only 60\,bit.
Additionally, as defined in the IEEE 802.11ac~\cite{11ac}, the subcarrier-averaged stream gain $\bar{\bm{\varLambda}}$ are quantized with a quantization step size of 0.25\,dB.

Moreover, to assess the error of multiple AoD estimations, we swap the order of the estimated AoDs to minimize the error between the estimated AoDs and ground-truth AoDs; subsequently, the error is calculated from the swapped versions of the estimated and ground-truth AoDs.
The detailed parameters are as follows: the distance of each antenna element is 25\,mm, the number of subcarriers is 52, the bandwidth is 20\,MHz, the center frequency is 5.18\,GHz, the number of CSIs used for each AoD estimation $N^\mathrm{pct}$ is ten, and the number of antenna elements in each sub-antenna array $M'$ is two.

\begin{figure}[!t]
    \centering
    \includegraphics[width=0.4\textwidth]{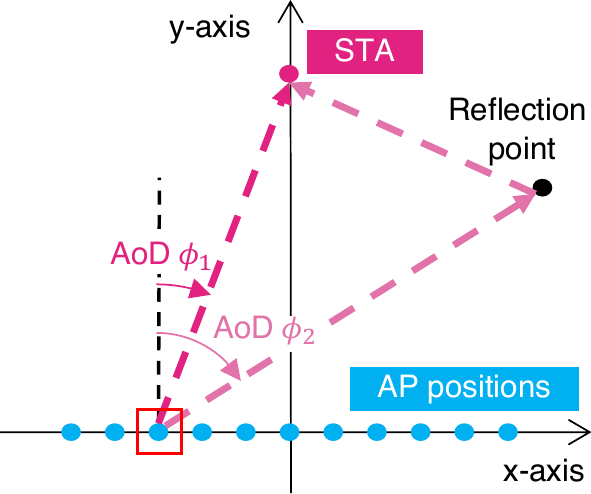}
    \caption{Numerical evaluation environment.
        The STA and reflection point exist at (0\,m, 10\,m) and (5.5\,m, 3\,m) in a two-dimensional space, respectively.
        The AP exists at either of 11 points denoted by blue dots.
        Color-dots lines indicate two propagation paths when AP exists on a point surrounded by a red square.
    }
    \label{fig:setup}
\end{figure}

\subsection{Result}
Fig.~\ref{fig:neu_result_music} shows an example of the MUSIC spectrum function $g(\phi)$ of the BFF- and CSI-based MUSIC algorithms, respectively.
The results denoted in Fig.~\ref{fig:neu_result_music} are obtained with the setting that the signal-to-noise ratio (SNR) is 20\,dB, and the AP exists at ($-3$\,m, 0\,m), which is the AP position surrounded by the red square in Fig.~\ref{fig:setup}.
The two peaks of the MUSIC spectrum function indicate the two estimated AoDs.
The estimated AoDs of the BFF-based MUSIC match with the ground-truth AoDs, as well as that of the CSI-based MUSIC.

Table~\ref{tab:median_nume} shows the median of the absolute error of the AoD estimation by the CSI- and BFF-based MUSIC for each SNR.
Regardless of the SNR, the error of the CSI-based MUSIC is lower than or equivalent to that of the BFF-based MUSIC.
This is because the BFF is highly quantized; specifically, the $4\times 4$ right singular matrix is represented by only 60\,bit.
However, the difference in the error between the two sensing methods is trivial.
Specifically, to estimate the AoDs of the direct and indirect paths, the difference is smaller than 0.03\textdegree\,and 0.4\textdegree, respectively.
Thus, we can conclude that the BFF-based MUSIC accurately estimates multiple AoDs; moreover, the accuracy of the BFF-based MUSIC is comparable to that of the CSI-based MUSIC.
If the BFF is not quantized, the result of the AoD estimation from the CSI and BFF matches perfectly.

\begin{figure}[!t]
    \centering
    {\includegraphics[width=0.44\textwidth]{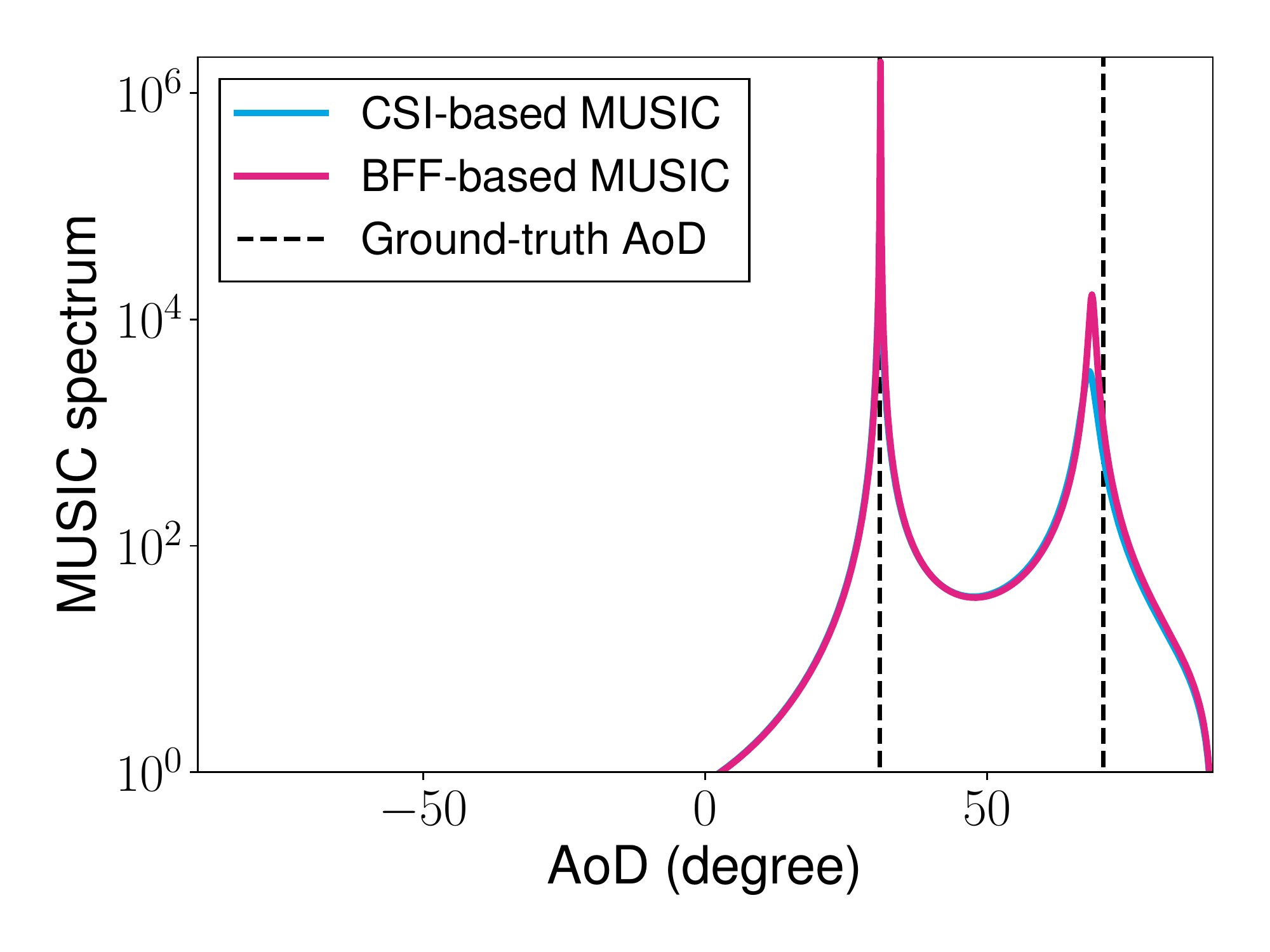}}
    \caption{MUSIC spectrum function of BFF- and CSI-based MUSIC obtained in numerical evaluation.
        The two peaks of the function indicate the two estimated AoDs.
    }
    \label{fig:neu_result_music}
\end{figure}

\begin{table}[t!]
    \caption{
        Median of absolute error of AoD estimation by CSI- and BFF-based MUSIC for each SNR.
    }
    \centering
    \scalebox{1.}{
        \begin{tabular}{cccccc}
            \toprule
            SNR    & \multicolumn{2}{c}{CSI} & \multicolumn{2}{c}{BFF}                                    \\
                   & Direction               & indirection             & Direction       & indirection    \\
                   & path                    & path                    & path            & path           \\
            \midrule
            5\,dB  & 0.11\textdegree         & 2.4\textdegree          & 0.13\textdegree & 2.8\textdegree \\
            10\,dB & 0.09\textdegree         & 1.0\textdegree          & 0.09\textdegree & 1.1\textdegree \\
            20\,dB & 0.06\textdegree         & 0.2\textdegree          & 0.09\textdegree & 0.3\textdegree \\
            \bottomrule
        \end{tabular}
    }
    \label{tab:median_nume}
\end{table}

\section{Experimental Evaluation}
\label{sec:exp_evaluation}
This study evaluated the accuracy of BFF- and CSI-based MUSIC algorithms in various real-world environments,
where the line-of-sight (LoS) path between the AP and STA exists.
Notably, this evaluation is based on the assumption that the number of propagation paths is one (i.e., only the direct path exists), and the ground-true AoD is defined as the AoD of the LoS path.
This assumption was adopted because we cannot measure the ground-truth AoDs of the reflection paths in the real-world environment.

Experimental evaluations were performed in three real-world scenarios: indoor, outdoor, and semi-outdoor scenarios.
The indoor, outdoor, and semi-outdoor scenarios differ in terms of the effect of the reflection paths.
Specifically, the received power caused by the reflection paths in the indoor scenario is generally larger than the outdoor and semi-outdoor scenarios.
The outdoor and semi-outdoor scenarios differ in terms of the method to vary AoD.
In the outdoor scenario, the position and orientation of the antenna array of the AP are fixed, and the AoD only depends on the position of the STA.
However, in the semi-outdoor scenario, the AP and STA are fixed, and the AoD only depends on the orientation of the AP's antenna array.

\subsection{Setup}
\noindent\textbf{Experimental equipment:}
The experimental system consists of an AP and STA equipped with three and two antennas, resulting in the $2\times 3$ CSI matrix.
As shown in Fig.~\ref{fig:AP_setup}, the antenna elements of the AP are linearly aligned, where the distance of the conservative antenna elements is 25\,mm.
The communication protocol, the wireless channel, the bandwidth, and the number of subcarriers are IEEE 802.11ac, 104ch, 20\,MHz, and 52, respectively.
Moreover, ASUS RT-AC86U is used for the AP and STA.
The detailed parameters of the MUSIC algorithm are as follows: the number of CSIs or BFFs used for each AoD estimation $N^\mathrm{pct}$ is ten, and the number of antenna elements in each sub-antenna array $M'$ is two.

\noindent\textbf{BFF estimation:}
Notably, to provide fair comparisons between CSI- and BFF-based sensing,
we used a firmware modification~\cite{gringoli2019free} to extract CSI from the AP and calculate BFF from the extracted CSI.
Specifically, assuming the channel reciprocity, we emulated the CSI measured at the STA as the transpose of the CSI measured at the AP.
From the CSI, the corresponding BFF is calculated following the IEEE 802.11ac standard as described in Section~\ref{sssec:BFF}.

Because the shape of CSI is $2\times 3$, the right singular matrix $\bm{V}_k$ is represented by 12 angles with the quantization step size $\Delta$.
Unless otherwise noted, this evaluation used $\Delta$ of $\pi/32$\,rad, resulting in a $2\times 3$ complex matrix ${\bm{V}_k}^\mathrm{H}$ represented by 30\,bits.\footnotemark[3].
Additionally, as defined in the IEEE 802.11ac~\cite{11ac}, the subcarrier-averaged stream gain $\bar{\bm{\varLambda}}$ were quantized with a quantization width of 0.25\,dB.

\begin{figure}[!t]
    \centering
    \includegraphics[width=0.3\textwidth,page = 4]{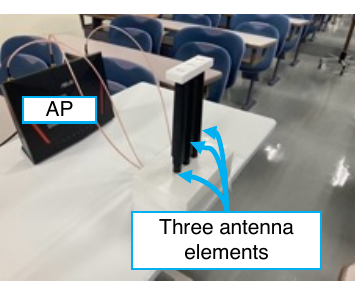}
    \caption{Snapshot of AP.
        Three antennas are linearly aligned with 25\,mm of the space between antennas.}
    \label{fig:AP_setup}
\end{figure}

\noindent\textbf{Experimental scenario:}
The experimental evaluation was performed on three scenarios: outdoor, semi-outdoor, and indoor scenarios.
An LoS path exists between the AP and STA in the three scenarios.
For all the scenarios, the CSIs and corresponding BFFs were obtained at multiple arrangements regarding the AP and STA, where the ground-truth AoD differs by the arrangement.
Regardless of the scenario, the AP captures approximately 850 packets from the STA at each equipment arrangement and estimates CSI and BFF for each captured packet.

\begin{figure}[!t]
    \centering
    \subfloat[Equipment layout.]{\includegraphics[width=0.25\textwidth, page=1]{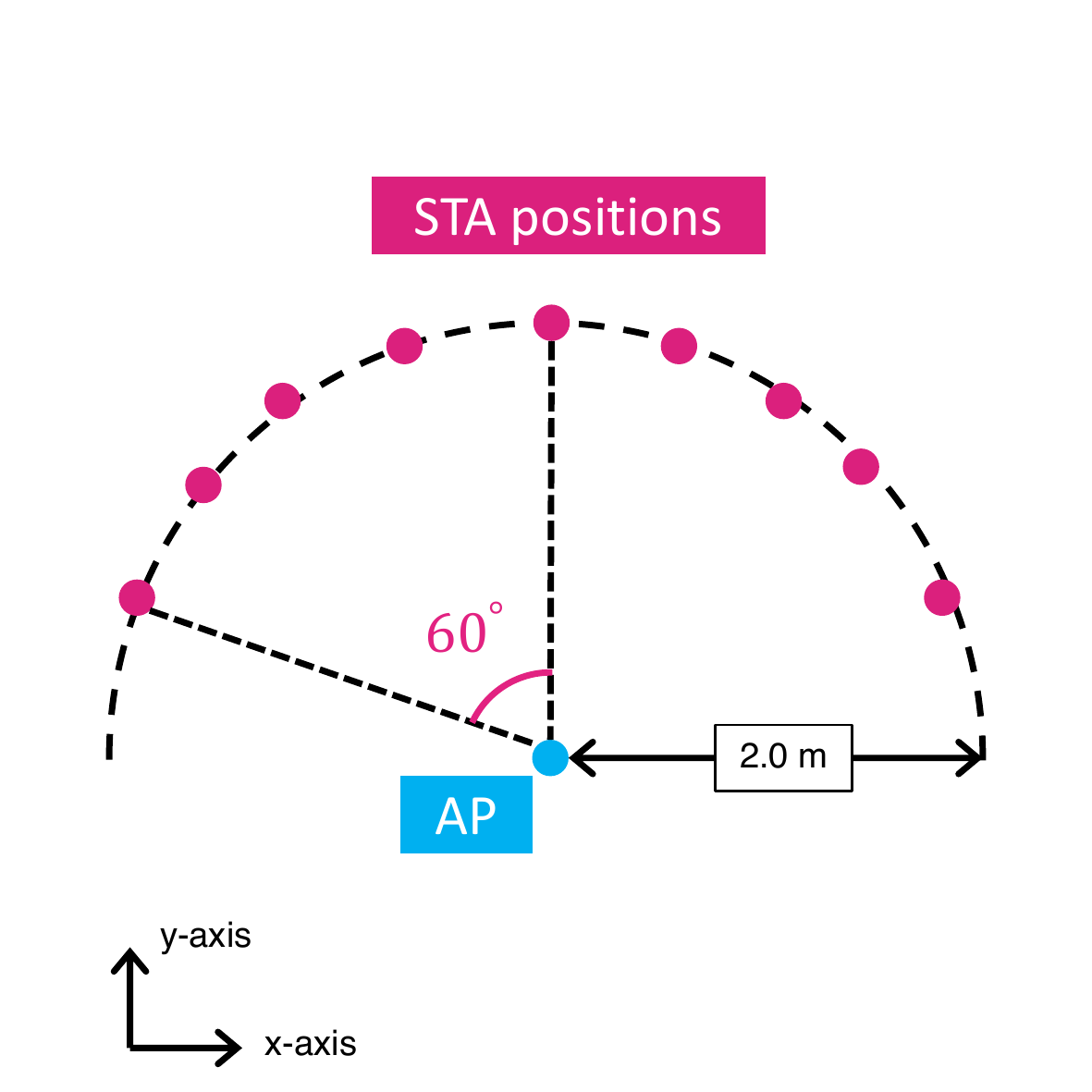}}
    \subfloat[Snapshot.]{\includegraphics[width=0.25\textwidth,page = 6]{photo.pdf}}\\
    \caption{Outdoor experimental scenario.
        STA is placed at either of the nine red points.
        AP and STA are located at a height of 0.9\,m.
    }
    \label{fig:outdoor_expsetup_A}
\end{figure}

\begin{figure}[!t]
    \centering
    \subfloat[Setup.]{\includegraphics[width=0.25\textwidth, page=2]{exp_setup.pdf}}
    \subfloat[Snapshot.]{\includegraphics[width=0.25\textwidth,page = 1]{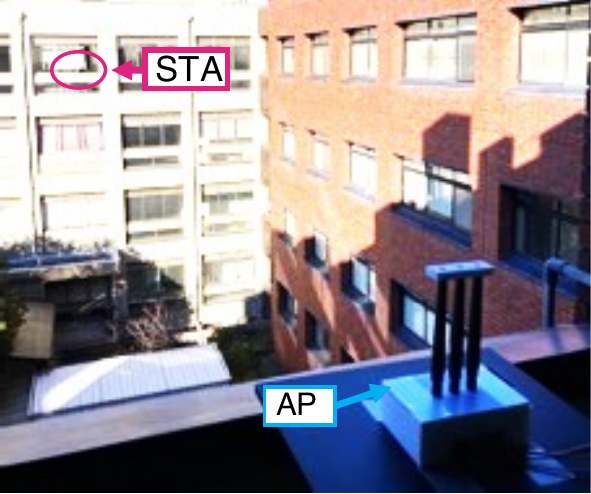}}\\
    \caption{Semi-outdoor experimental scenario.
        AP and STA are located in different rooms on the fourth floor, where the LoS path exists through open windows.
        The height of AP and STA from the floor is 0.9\,m, and that of the rooms is 3.0\,m.}
    \label{fig:outdoor_expsetup_B}
\end{figure}

\begin{figure}[!t]
    \centering
    \subfloat[Setup.]{\includegraphics[width=0.45\textwidth]{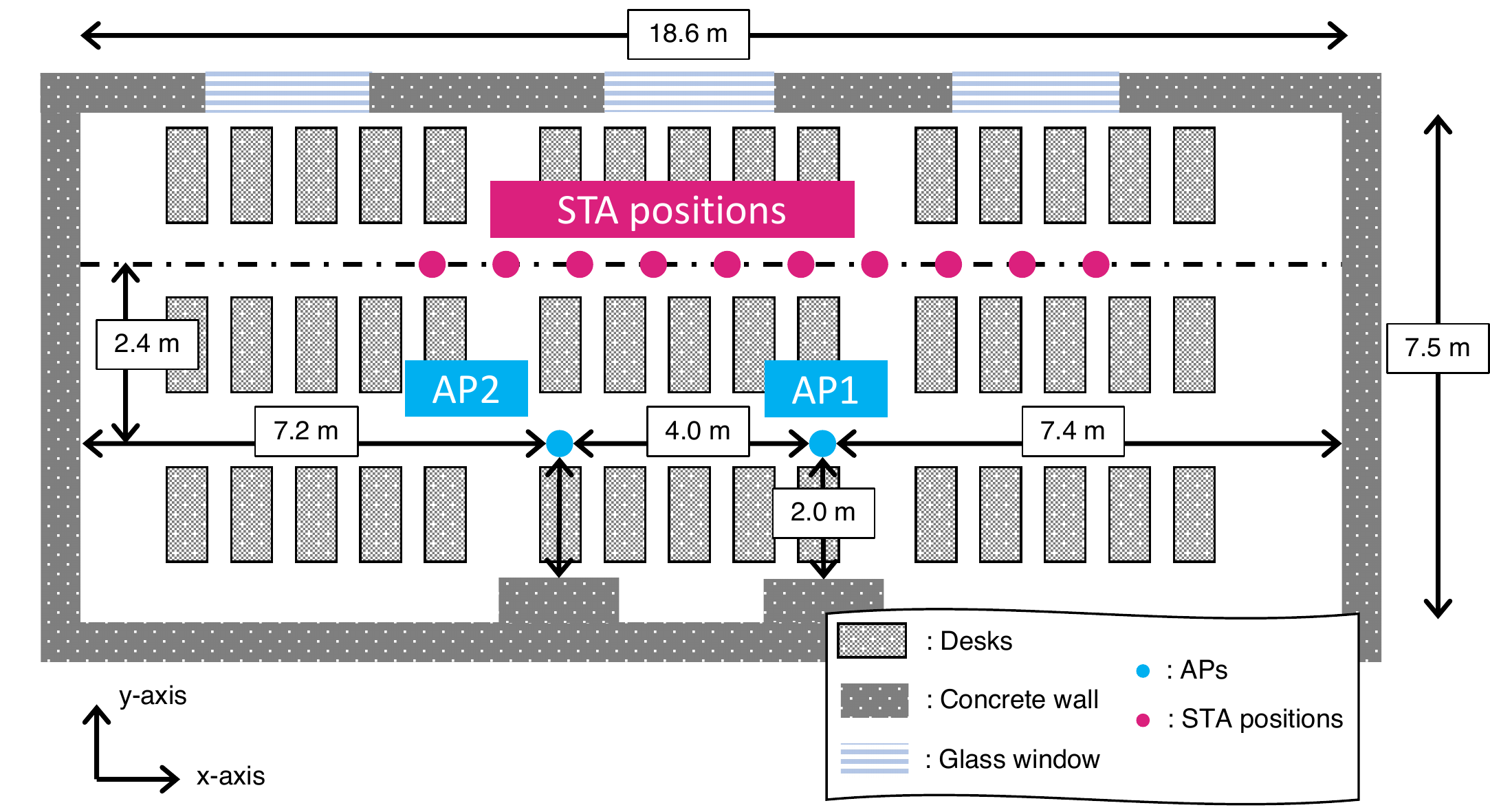}}\\
    \subfloat[Snapshot.]{\includegraphics[width=0.3\textwidth,page = 5]{photo.pdf}}\\
    \caption{Indoor experimental scenario.
        STA is placed at either of the ten red points, whereas two APs are located at the blue points.
        The height of AP and STA are 0.9\,m.
        The height, width, and depth of the room are 3.0\,m, 7.5\,m, and 18.6\,m, respectively.
    }
    \label{fig:indoor_expsetup}
\end{figure}

\begin{figure}[!t]
    \centering
    \subfloat[Wiring diagram.]{\includegraphics[width=0.2\textwidth]{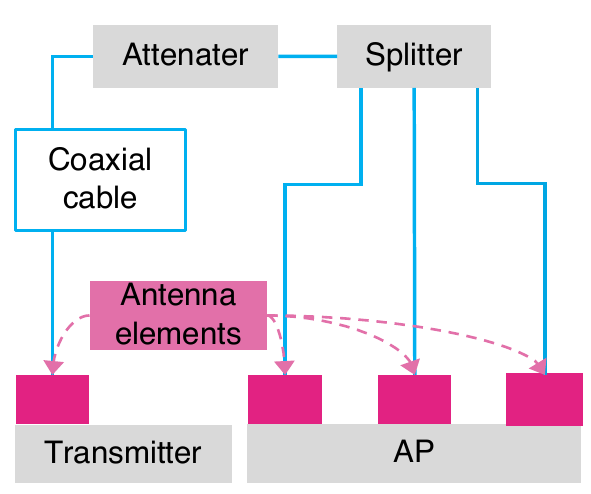}}
    \subfloat[Snapshot.]{\includegraphics[width=0.2\textwidth, page=2]{photo.pdf}}
    \caption{Setup of calibration procedure.
        The lengths of the coaxial cables are adjusted so that the phases at the three antennas of the AP are the same.}
    \label{fig:calib_setup}
\end{figure}

Fig.~\ref{fig:outdoor_expsetup_A} shows the setup and snapshot of the outdoor scenario.
The STA is placed at either of the nine positions on the circle with a radius of 2.0\,m centered on the AP.
The orientations of the antenna array of the AP and STA are fixed parallel to the x-axis.
Thus, the AoD only depends on the position of the STA, resulting in the AoD being either $-$60\textdegree\, to 60\textdegree\, in 15\textdegree\, increments.

Fig.~\ref{fig:outdoor_expsetup_B} shows the semi-outdoor experimental scenario and its snapshot.
The AP and STA are fixed in different rooms, where the LoS exists through the open windows.
In the semi-outdoor environment, the orientation of the antenna array of the AP is changed, whereas the orientation of the antenna array of the STA is fixed parallel to the y-axis.
Thus, the AoD only depends on the orientation of the AP's antenna array.
Specifically, the orientation of the AP is changed so that the AoD is either of $-$60\textdegree\, to 60\textdegree\, in 15\textdegree\, increments.

Figs.~\ref{fig:indoor_expsetup} shows the indoor experimental scenario and its snapshot.
The two APs and STA are located in a lecture room, where the orientation of the antenna array of the APs and STA is fixed parallel to the y-axis.
While the APs are fixed, the STA is located at either of ten positions on a line parallel to the x-axis, where the distance between the line and AP is 2.4\,m.
Thus, the AoD only depends on the STA's position.
In the scenario, the AoD is varied from approximately $-$60\textdegree\, to 60\textdegree.
It should be noted that the AoD estimation is conducted for each AP.

\noindent\textbf{Calibration procedure:}
Fig.~\ref{fig:calib_setup} shows the setup of the calibration procedure.
The AP's antennas and a transmitter antenna are connected via coaxial cables.
Because the length of the coaxial cables between the antenna of the AP and the transmitter are the same among the three antennas of the AP,
the phase of the AP's antennas are considered to be the same and there exists only a direct wave (i.e., $L=1$ and $\hat{\phi}=0$).
We captured approximately 1{,}000 packets in the environment, obtained CSIs, and calculated BFFs.
From the BFFs, we estimated the calibration matrix $\bm{W}$ as detailed in Section~\ref{ssec:calib}.

\subsection{Results}

\begin{figure}[!t]
    \centering
    \includegraphics[width=0.4\textwidth]{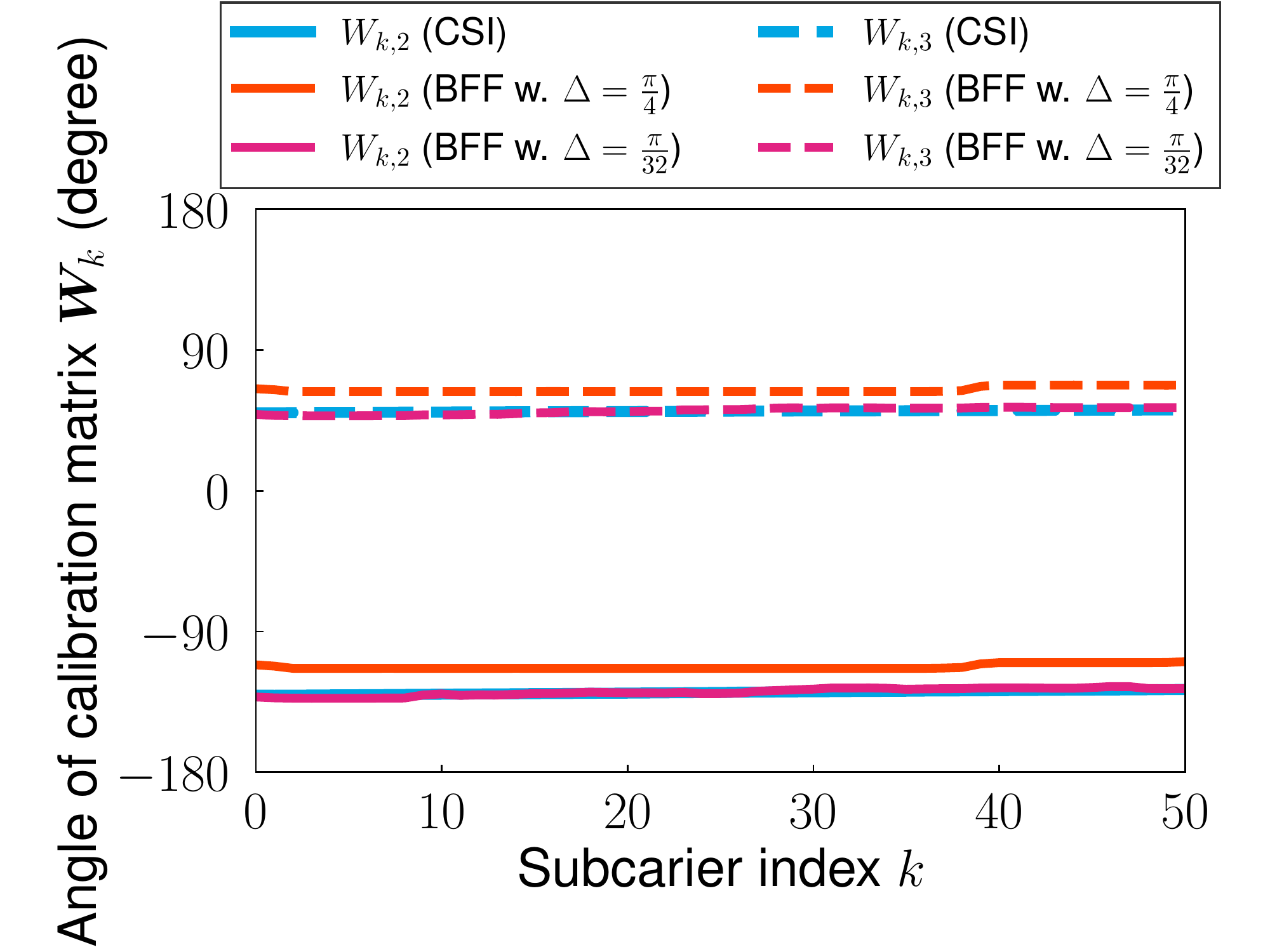}
    \caption{Angle of estimated calibration matrix $\bm{W}_k$ from BFF and CSI, respectively.
        Calibration matrix $\bm{W}_{k}$ is represented by $\mathrm{diag}(1,W_{k,2},W_{k,3})$.
    }
    \label{fig:cali_result}
\end{figure}

\textbf{Results of calibration procedure:}
Fig.~\ref{fig:cali_result} shows the angle of the estimated calibration matrix $\bm{W}_k$ from the BFF with a quantization step size of $\pi/32$\,rad and $\pi/4$\,rad, and CSI, respectively.
As denoted in Section~\ref{ssec:calib}, the calibration matrix $\bm{W}_k$ is denoted as $\mathrm{diag}(1,W_{k,2},W_{k,3})$.
Thus, Fig.~\ref{fig:cali_result} depicts the argument of $W_{k,2}$ and $W_{k,3}$, respectively.
When the quantization step size is $\pi/32$\,rad, the estimated arguments from the BFF match that of the CSI;
Specifically, the difference between the arguments estimated from the BFF and CSI is smaller than 2.3\textdegree\, regardless of the subcarrier index.
Thus, we can conclude that when the quantization step size is small, the results of the BFF-based calibration accurately matches that of the CSI-based calibration.

As the quantization step size is increased,
the difference in the estimated arguments between the BFF and CSI increases because of the quantization error induced in the BFF.
Specifically, when the quantization step size is $\pi/4$\,rad, the median and maximum difference between the estimated arguments from the BFF and that from CSI is 15.4\textdegree\,and 25.0\textdegree, respectively.
However, the following evaluations reveal that, even when the quantization step size is large, the BFF-based MUSIC achieved comparable AoD estimation accuracy to the CSI-based method.

\begin{table}[t!]
    \caption{
        Median of absolute error of AoD estimation by CSI- and BFF-based MUSIC for each scenario.
    }
    \centering
    \scalebox{1.}{
        \begin{tabular}{ccc}
            \toprule
            Scenario      & CSI             & BFF             \\
            \midrule
            Outdoor       & 10.3\textdegree & 10.3\textdegree \\
            Semi-outdoor  & 10.2\textdegree & 9.8\textdegree  \\
            Indoor w. AP1 & 17.9\textdegree & 17.1\textdegree \\
            Indoor w. AP2 & 14.7\textdegree & 14.8\textdegree \\
            \bottomrule
        \end{tabular}
    }
    \label{tab:median}
\end{table}

\begin{figure}[!t]
    \centering
    \subfloat[Outdoor scenario.]{\includegraphics[width=0.23\textwidth]{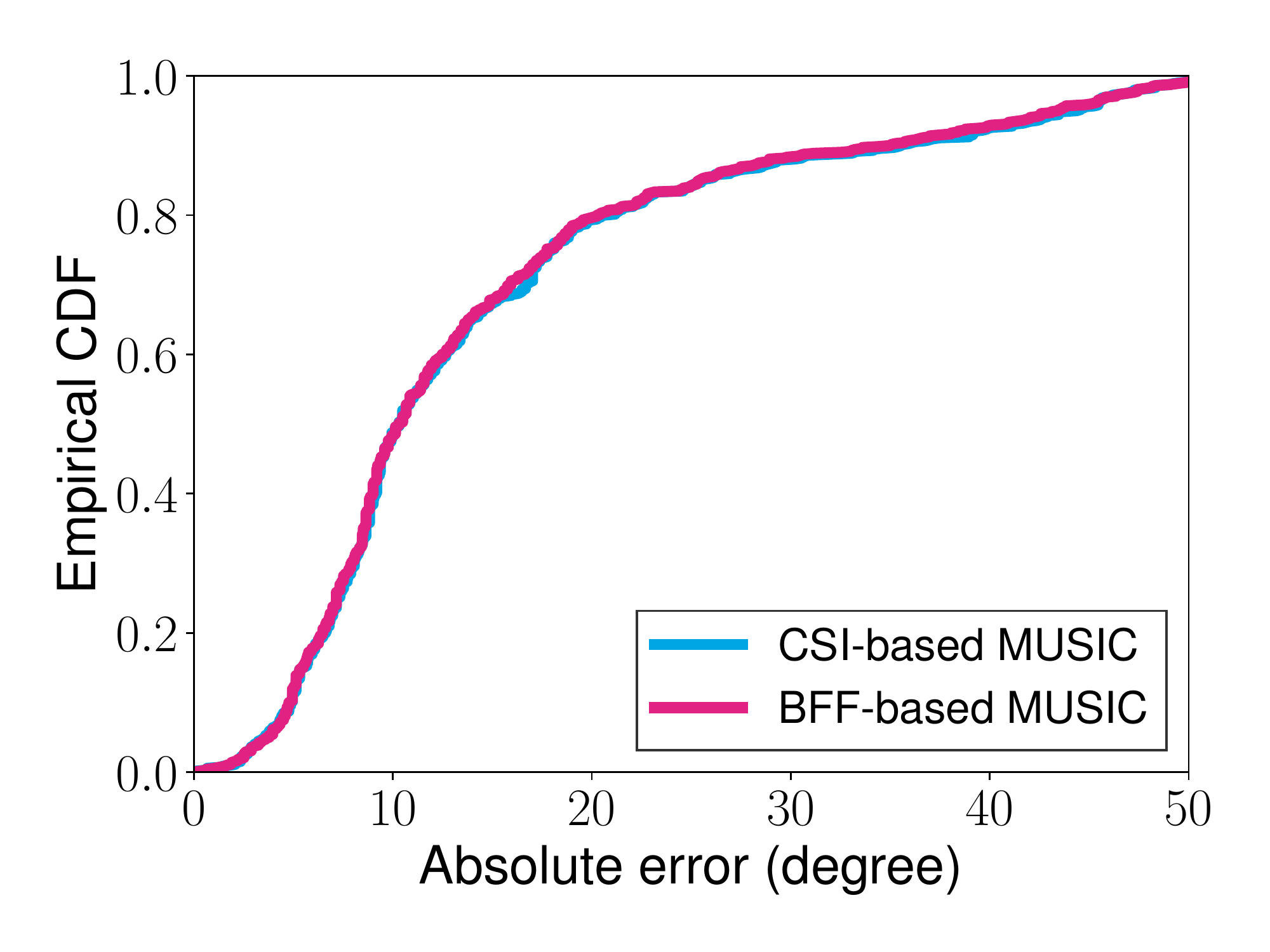}}
    \subfloat[Semi-outdoor scenario.]{\includegraphics[width=0.23\textwidth]{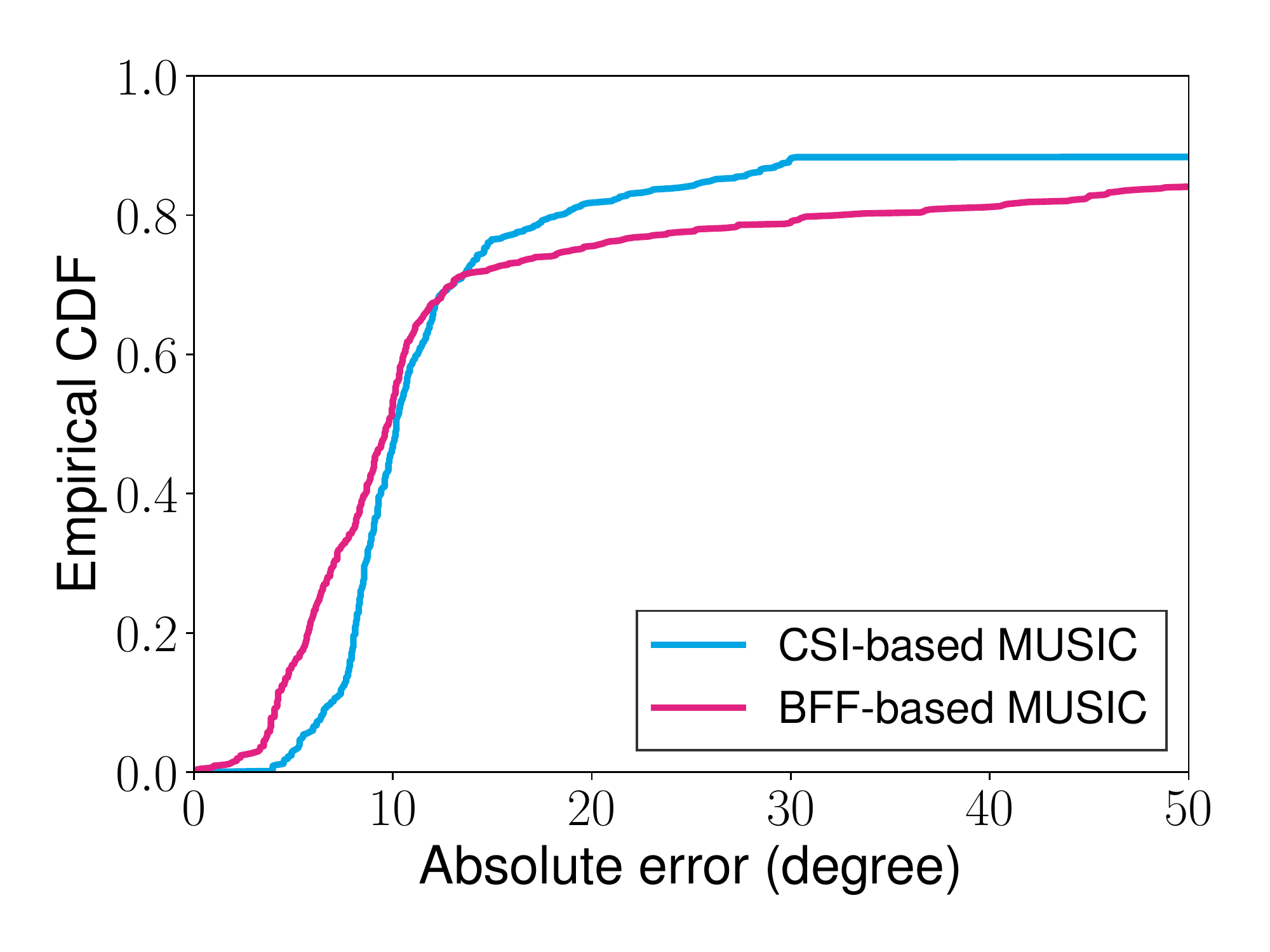}}\\
    \subfloat[Indoor scenario w. AP1.]{\includegraphics[width=0.23\textwidth]{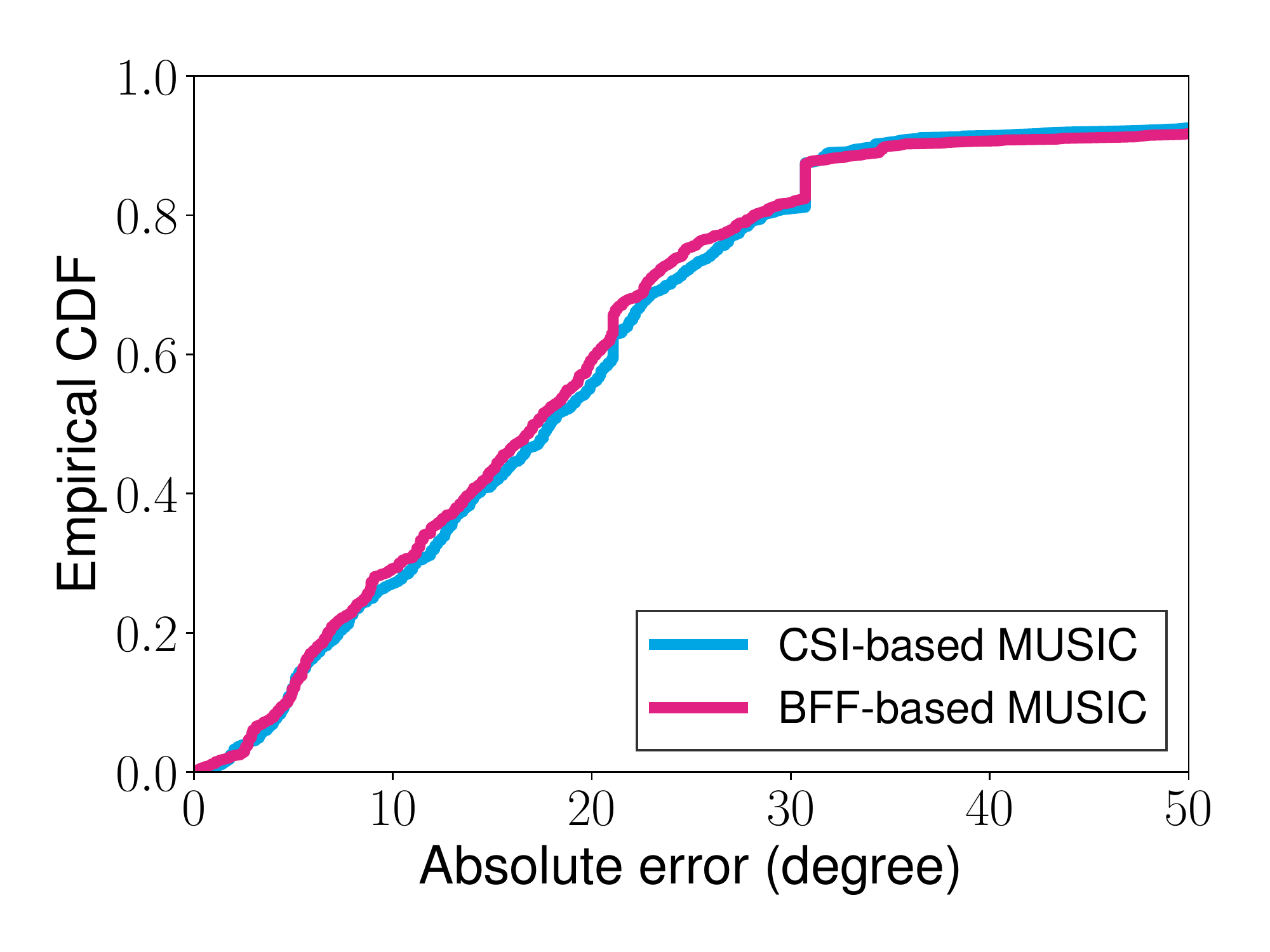}}
    \subfloat[Indoor scenario w. AP2.]{\includegraphics[width=0.23\textwidth]{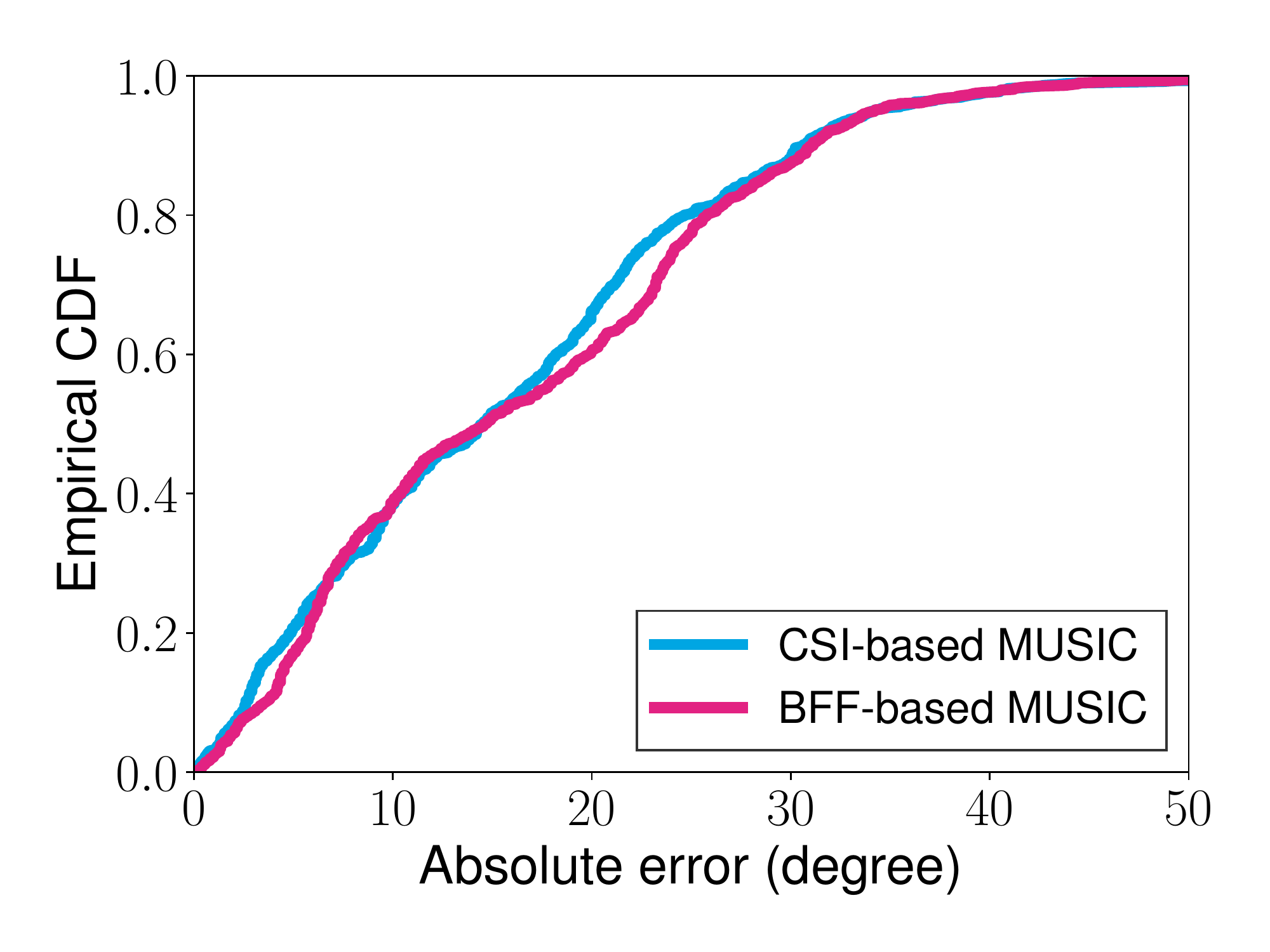}}
    \caption{Empirical CDF of absolute error of AoD estimation by CSI-based MUSIC and BFF-based MUSIC for each scenarios.}
    \label{fig:CDFs}
\end{figure}

\textbf{AoD estimation error comparison:}
Fig.~\ref{fig:CDFs} shows the empirical cumulative distribution function (CDF) of the AoD estimation error using BFF- and CSI-based MUSIC, respectively.
In Fig.~\ref{fig:CDFs}, the error of the BFF-based MUSIC is comparable to that of the CSI-based MUSIC regardless of the experimental scenarios.
Table~\ref{tab:median} lists the error medians of the AoD estimation by CSI-based MUSIC and BFF-based MUSIC in the three scenarios.
As shown in Fig.~\ref{fig:CDFs}, regardless of the scenario, the errors of the BFF- and CSI-based MUSIC are comparable.
Thus,  the BFF-based MUSIC achieves comparable AoD estimation accuracy to CSI-based MUSIC,
although the BFF is highly quantized, specifically, the $2\times 3$ right singular matrix is represented by only 30\,bit, and the subcarrier-averaged stream gain are represented with a quantization step size of 0.25\,dB.

Additionally, the error of the CSI-based MUSIC in this evaluation is comparable to the previously reported value~\cite{tian2016pila}, that is approximately 10\textdegree.
Although the error of the AoD estimation highly depends on the experimental environments and equipment (e.g., the antenna characteristics, the propagation environment, and the placement of the AP and STA),
the similarity of the error between this paper and the existing report~\cite{tian2016pila} indicates that the implementation in this study is adequate.

Upon comparing the estimation errors between the scenarios, the errors of the indoor scenarios are found to be higher than those of the outdoor and semi-outdoor scenarios for both the CSI- and BFF-based MUSIC.
This is because the number of propagation paths at the indoor scenario is larger than that at the outdoor and semi-outdoor scenarios.
Because we assumed $L=1$ in this experimental evaluation, the larger multipath degrades the AoD estimation accuracy.

\textbf{Impact of quantization step size:}
Fig~\ref{fig:quantize} shows the impact of the quantization step size on the AoD estimation error of the BFF-based MUSIC.
In IEEE 802.11ac~\cite{11ac}, four quantization step sizes $\Delta$ of $\bm{V}_k$ are defined, namely $\pi/4$\,rad, $\pi/8$\,rad, $\pi/16$\,rad, and $\pi/32$\,rad.
Regardless of the experimental scenarios, the impact of the quantization step size on the median of error is less than 3.0\textdegree.
Moreover, regardless of the experimental scenarios and the quantization step size, the BFF-based MUSIC achieved a comparable AoD estimation error to the CSI-based methods.
Thus, even when the AP adopts the largest quantization step size defined in IEEE 802.11ac, the BFF-based MUSIC achieves comparable AoD estimation accuracy to CSI-based MUSIC.

\begin{figure}[!t]
    \centering
    \subfloat[Outdoor and semi-outdoor scenario.]{\includegraphics[width=0.23\textwidth]{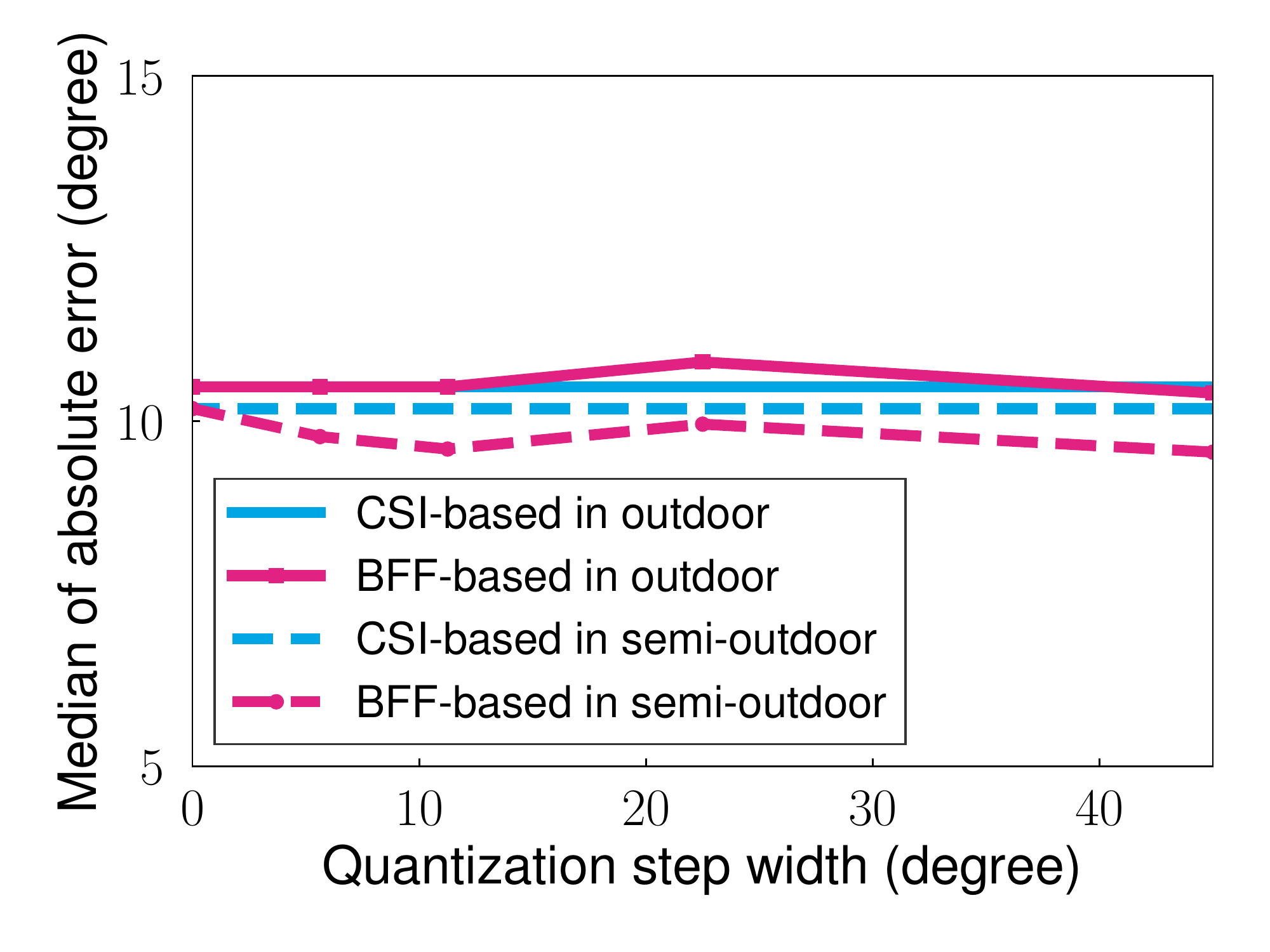}}
    \subfloat[Indoor scenario.]{\includegraphics[width=0.23\textwidth]{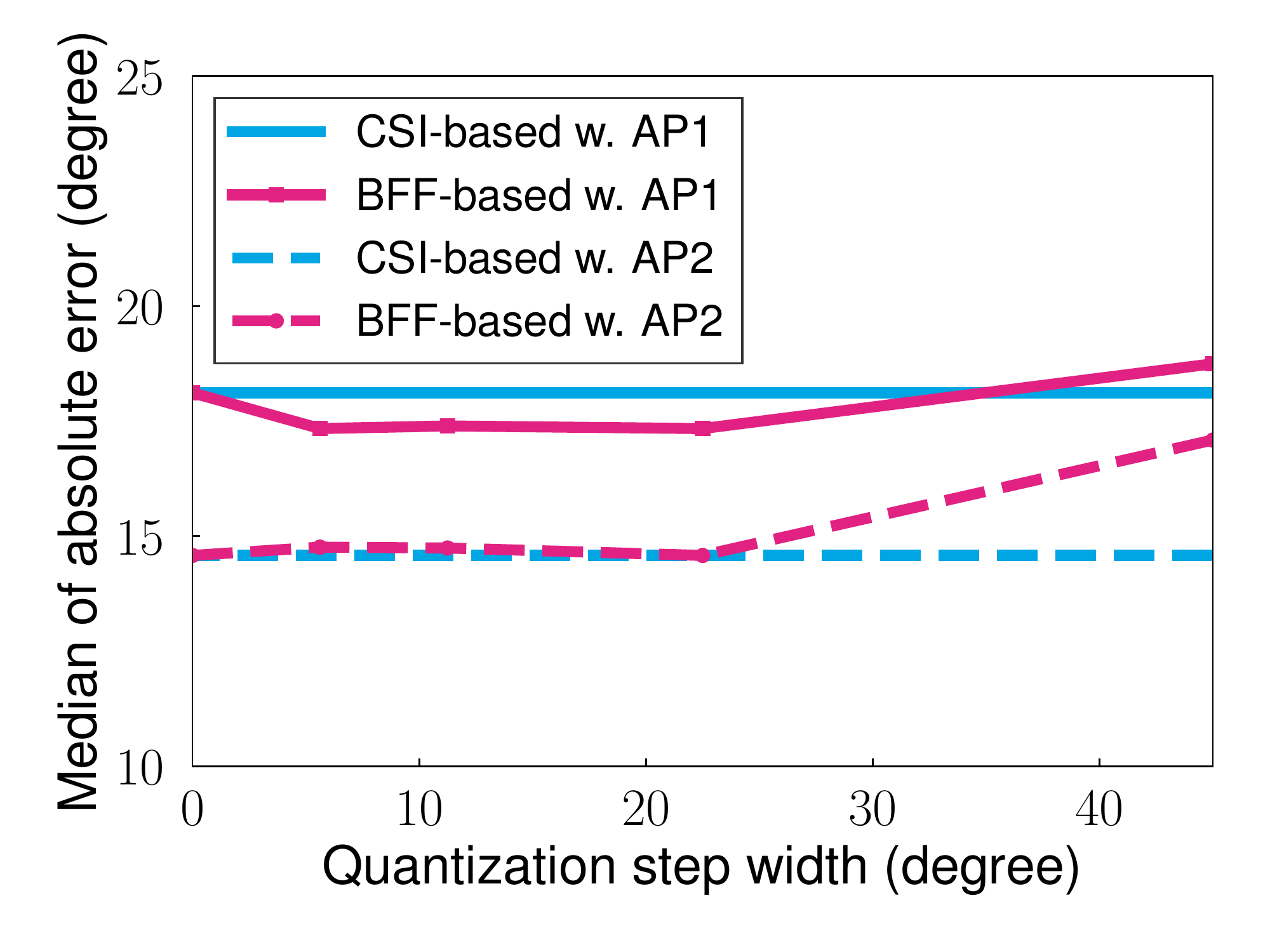}}
    \caption{Impact of quantization step size on median error of AoD estimation for each scenario.}
    \label{fig:quantize}
\end{figure}

\section{Conclusion}
\label{sec:conclusion}
This study confirmed that, to estimate multiple AoDs, an extension of the MUSIC algorithm is applicable using BFF, which contains only subcarrier-averaged stream gain and the highly quantized right singular matrix.
Numerical and experimental evaluations on three scenarios revealed that the AoD estimation accuracy of BFF-based MUSIC is comparable to that of CSI-based MUSIC.

\bibliographystyle{IEEEtran}
\bibliography{ith.bbl}

\begin{IEEEbiography}
    [{\includegraphics[width=1in, height=1.25in, clip, keepaspectratio]{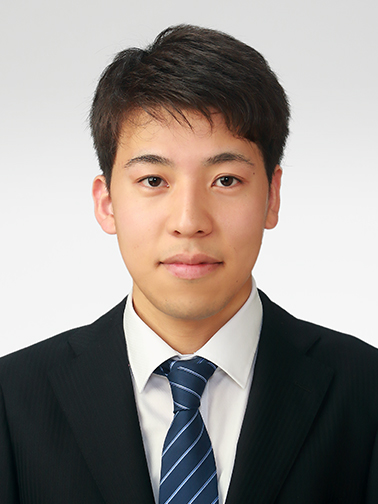}}]
    {Sohei~Itahara}
    received the B.E. degree in electrical and electronic engineering from Kyoto University in 2020.
    He is currently studying toward the M.I. degree at the Graduate School of Informatics, Kyoto University.
    He is a student member of the IEEE.
\end{IEEEbiography}

\begin{IEEEbiography}
    [{\includegraphics[width=1in, height=1.25in, clip, keepaspectratio]{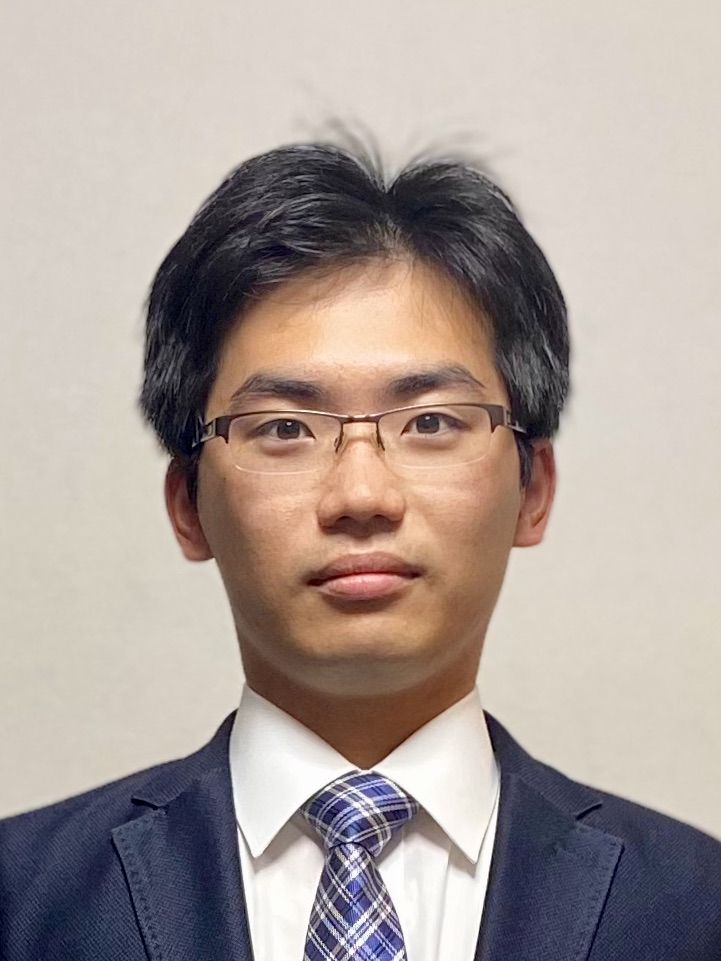}}]
    {Sota~Kondo}
    received the B.E. degree in electrical and electronic engineering from Kyoto University in 2021.
    He is currently studying toward the M.I. degree at the Graduate School of Informatics, Kyoto University.
    He is a student member of the IEEE.
\end{IEEEbiography}

\begin{IEEEbiography}
    [{\includegraphics[width=1in, height=1.25in, clip, keepaspectratio]{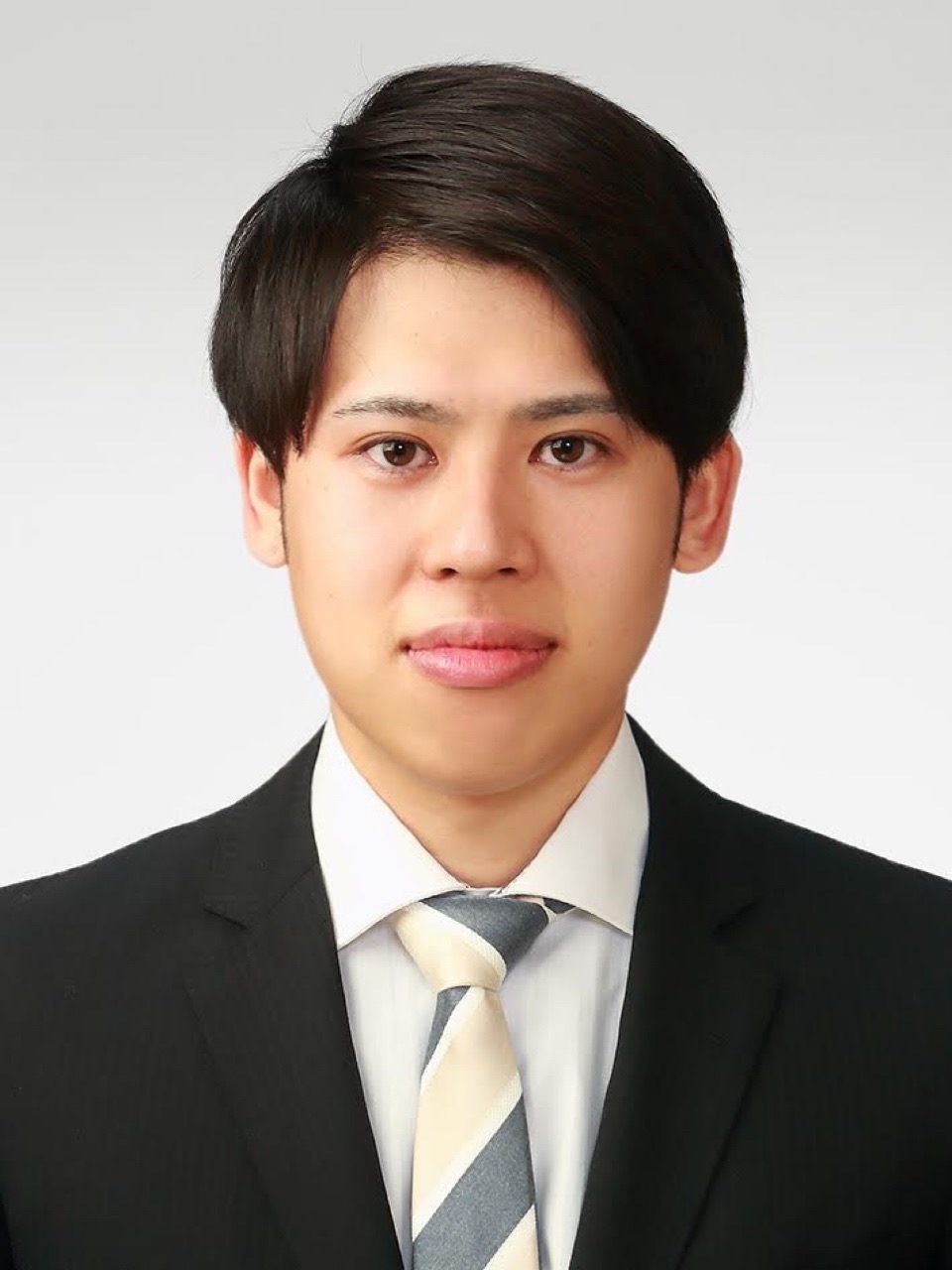}}]
    {Kota~Yamashita}
    received the B.E. degree in electrical and electronic engineering from Kyoto University in 2020.
    He is currently studying toward the M.I. degree at the Graduate School of Informatics, Kyoto University.
    He is a student member of the IEEE.
\end{IEEEbiography}

\begin{IEEEbiography}
    [{\includegraphics[width=1in, height=1.25in, clip, keepaspectratio]{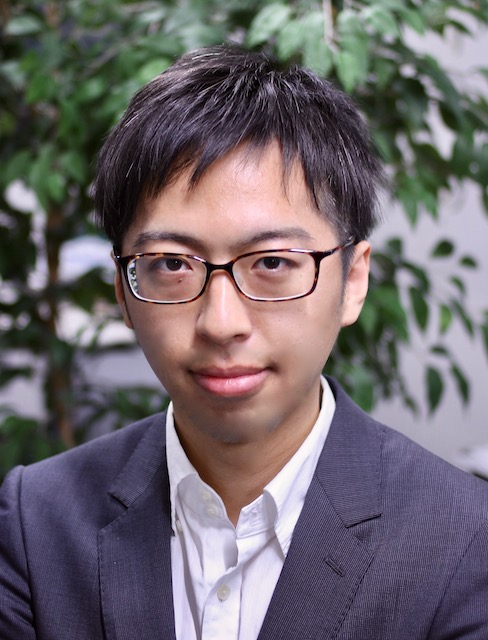}}]
    {Takayuki~Nishio}
    has been an associate professor in the School of Engineering, Tokyo Institute of Technology, Japan, since 2020.
    He received the B.E.\ degree in electrical and electronic engineering and the master's and Ph.D.\ degrees in informatics from Kyoto University in 2010, 2012, and 2013, respectively.
    He had been an assistant professor in the Graduate School of Informatics, Kyoto University from 2013 to 2020.
    From 2016 to 2017, he was a visiting researcher in Wireless Information Network Laboratory (WINLAB), Rutgers University, United States.
    His current research interests include machine learning-based network control, machine learning in wireless networks, and heterogeneous resource management.
\end{IEEEbiography}

\begin{IEEEbiography}
    [{\includegraphics[width=1in, height=1.25in, clip, keepaspectratio]{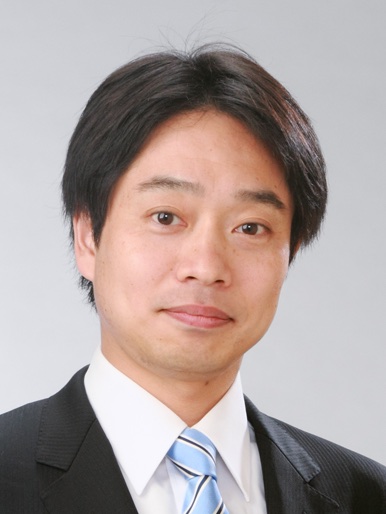}}]
    {Koji~Yamamoto}
    received the B.E.\ degree in electrical and electronic engineering from Kyoto University in 2002, and the master and Ph.D.\ degrees in Informatics from Kyoto University in 2004 and 2005, respectively.
    From 2004 to 2005, he was a research fellow of the Japan Society for the Promotion of Science (JSPS).
    Since 2005, he has been with the Graduate School of Informatics, Kyoto University, where he is currently an associate professor.
    From 2008 to 2009, he was a visiting researcher at {Wireless@KTH}, Royal Institute of Technology (KTH) in Sweden.
    He serves as an editor of IEEE Wireless Communications Letters, IEEE Open Journal of Vehicular Technology, and Journal of Communications and Information Networks, a symposium co-chair of GLOBECOM 2021, and a vice co-chair of IEEE ComSoc APB CCC.
    He was a tutorial lecturer in IEEE ICC 2019.
    His research interests include radio resource management, game theory, and machine learning.
    He received the PIMRC 2004 Best Student Paper Award in 2004, the Ericsson Young Scientist Award in 2006.
    He also received the Young Researcher's Award, the Paper Award, SUEMATSU-Yasuharu Award, Educational Service Award from the IEICE of Japan in 2008, 2011, 2016, and 2020, respectively, and IEEE Kansai Section GOLD Award in 2012.
    He is a senior member of the IEEE and a member of the Operations Research Society of Japan.
\end{IEEEbiography}

\begin{IEEEbiography}[{\includegraphics[width=1in, height=1.25in, clip, keepaspectratio]{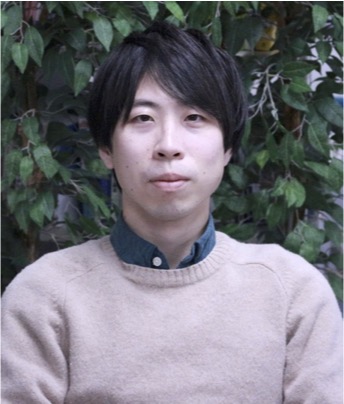}}]{Yusuke~Koda}
    received the B.E. degree in electrical and electronic engineering from Kyoto University in 2016, and the M.E. and the Ph.D. degree in informatics from the Graduate School of Informatics, Kyoto University in 2018 and 2021, respectively. He is currently a Postdoctoral Researcher with the Centre for Wireless Communications, University of Oulu, Finland, where he visited the Centre for Wireless Communications in 2019, to conduct collaborative research. He received the VTS Japan Young  Researcher's Encouragement Award in 2017 and TELECOM System Technology Award in 2020. He was a recipient of the Nokia Foundation Centennial Scholarship in 2019.
\end{IEEEbiography}

\EOD
\end{document}